%% file: algorithmic-regularity-full.tex
\documentclass[final, 11pt]{article}

\def\showauthornotes{0}
\def\showkeys{0}
\def\showdraftbox{0}

\input{macros}

\newtheorem{prop}[theorem]{Proposition}

\newtheorem{dfn}[theorem]{Definition}
\newtheorem{note}[theorem]{Note}

\newtheorem{remark}[theorem]{Remark}

\newtheorem{examplea}[theorem]{Example}

\def\rank{{\rm{rank}}}
\def\deg{{\rm{deg}}}
\def\bias{{\rm{bias}}}

\def\tP{\widetilde{P}}
\renewcommand{\hat}{\widehat}
\def\cF{\mathcal{F}}
\def\dim{{\mathrm{dim}}}
\def\tsigma{\widetilde{\sigma}}
\def\deg{{\mathrm{deg}}}
\def\tQ{\widetilde{Q}}
\def\rO{{\rm O}}
\def\tmu{\widetilde{\mu}}
\def\obs{\mathrm{obs}}
\DeclareMathOperator*{\argmin}{arg\,min}
\def\cC{\mathcal{C}}
\def\Der{\mathrm{Der}}

\def\tgamma{\widetilde{\gamma}}
\def\tsigma{\widetilde{\sigma}}
\def\tcF{\widetilde{\cF}}

\newcommand\ignore[1]{}

\newcommand{\restate}[2]{\medskip
\noindent{\bf #1 (restated).}{\sl #2}}

\title{%
Algorithmic regularity for polynomials and applications
}%

\author{
 Arnab Bhattacharyya\thanks{Indian Institute of Science.  Email: \texttt{arnabb@csa.iisc.ernet.in}.}
 \and
 Pooya Hatami\thanks{ University of Chicago. Email: \texttt{pooya@cs.uchicago.edu}.}
 \and
 Madhur Tulsiani\thanks{TTI Chicago.  Email: \texttt{madhurt@ttic.edu}. Research supported by NSF
    Career Award CCF-1254044.}
}

\date{\today}

\begin{document}

\sloppy

\maketitle

\draftbox

\begin{abstract}
In analogy with the regularity lemma of \Szemeredi \cite{Szem75},
regularity lemmas for polynomials shown by Green  and Tao \cite{GT09} and by Kaufman and Lovett \cite{KL08} give a way of modifying a given collection of
polynomials $\cF  = \{P_1,\ldots,P_m\}$ to a new collection $\cF'$ so that the polynomials in
$\cF'$ are ``pseudorandom''.  These lemmas have various applications, such as (special cases) of
Reed-Muller testing and worst-case to average-case reductions for polynomials. However, the
transformation from $\cF$ to $\cF'$ is not algorithmic for either regularity lemma.
We define new notions of regularity for polynomials, which are
analogous to the above, but which allow for
an efficient algorithm to compute the pseudorandom collection
$\cF'$. In particular, when the field is of high characteristic, in
polynomial time, we can refine $\cF$ into $\cF'$ where  every nonzero 
linear combination of polynomials in $\cF'$ has desirably small Gowers norm.


\medskip
Using the algorithmic regularity lemmas, we show that if a polynomial
$P$ of degree $d$ is within (normalized) Hamming distance
$1-\frac{1}{|\F|} 
-\eps$ of some unknown polynomial of degree $k$ over a prime field
$\F$ (for $k < d < |\F|$), then there is an
efficient algorithm for finding a degree-$k$ polynomial $Q$, which is within distance $1-\frac{1}{|\F|}
-\eta$ of $P$, for some $\eta$ depending on $\eps$. This can be thought of as decoding the
Reed-Muller code of order $k$ \emph{beyond} the list decoding radius, in the sense of finding one
close codeword, when the received word $P$
itself is a polynomial (of degree larger than $k$ but smaller than $|\F|$). 

\medskip
We  also obtain an algorithmic version of the worst-case to average-case reductions by Kaufman and Lovett
\cite{KL08}. They show that if a polynomial of degree $d$ can be weakly approximated by a polynomial of
lower degree, then it can be computed exactly using a collection of polynomials of degree at most
$d-1$. We give an efficient (randomized) algorithm to find this collection.
\end{abstract}

\thispagestyle{empty}
\newpage

\setcounter{page}{1}


\section{Introduction}
Regularity is a notion of ``pseudorandomness'' that allows one to decompose a given object
into a collection of simpler objects which appear random according to
certain statistics. The famous regularity lemma of
\Szemeredi \cite{Szem75, Szem78} says that any dense graph can be partitioned into a collection of
bounded number of ``pseudorandom'' bipartite graphs. The \Szemeredi regularity lemma has numerous
applications in combinatorics and property testing. 

The original proof by \Szemeredi was
non-algorithmic and the question of finding an algorithm for computing the regularity partition was
first considered by Alon \etal \cite{AlonDLRY94}, motivated (at least partly) by the problem of
converting some of the applications of the regularity lemma into algorithms. The regularity lemma is
often used to guarantee the existence of certain structures in a graph, and an algorithmic version
of the lemma allows one to \emph{find} these structures in a given graph. Since then, there have
been numerous improvements and extensions to the algorithmic version of the \Szemeredi regularity
lemma, of which the works \cite{AlonN06, FriezeK99, KohayakawaRT03, FischerMS07} constitute a
partial list (see \cite{FischerMS07} for a detailed discussion).

In studying a special case of the inverse conjecture for Gowers' norms over finite fields, Green and
Tao \cite{GT09} introduced a notion of regularity for a collection of polynomials. We call a
collection $\cF = \{P_1,\ldots,P_m\}$ of a bounded number of polynomials over  $\F^n$ 
for  a prime field $\F$, a {\deffont factor} of degree $d$, if all the polynomials in the
collection have degree at most $d$. The Green-Tao notion of regularity states that a given factor is
regular if every linear combination $\sum_{i=1}^m c_i \cdot P_i$ of the polynomials in the given
factor has high rank \ie if $k \leq d$ is the degree of the highest-degree polynomial with a
non-zero coefficient in the linear combination $Q = \sum_{i} c_i \cdot P_i$, then $Q$ cannot be
written as a function $\Gamma$ of some bounded number $M$ of degree-$(k-1)$ polynomials. Their
notion in fact allows $M$ to be a \emph{function of the number of polynomials in} $\cF$ (which is
$m$ here).

Green and Tao show that any given factor $\cF$ of bounded degree-$d$ with the number of polynomials
$m$ denoted as $\dim(\cF)$, and a function $F: \N \to \N$, one can ``refine'' it to a factor $\cF'$
of degree-$d$ which is $F$-regular \ie it is regular in the above sense, with the lower bound on the
rank given by $F(\dim(\cF'))$. Here by $\cF'$ being a refinement of $\cF$, we mean that each
polynomial $P$ in $\cF$ can be expressed as a function of polynomials from $\cF'$. Also, they show that
 the number of polynomials in $\cF'$ depends only on $d, \dim(\cF)$ and the function $F$, but is
 independent of the number of variables $n$. This can again be compared with the \Szemeredi
 regularity lemma, which is proved by showing that starting from any partition of a given graph,
 one can ``refine'' it to a regular partition, such that the number of pieces only depend on the
 regularity parameter and the number of pieces in the original partition.

\newcommand{\eF}[1]{e_{\F}(#1)}
We consider (and solve) the question of finding an algorithmic version of the above regularity
lemma, motivated by one of the applications in \cite{GT09}. 
For $c \in \F$, let $\eF{c}$ denote $e^{2\pi ic/|\F|}$, where $i$ is the square root of
$-1$. 
For $f: \F^n \to \C$, let $\Delta_h(f)(x) =
f(x+h)\overline{f(x)}$ denote the multiplicative \emph{derivative} of $f$ in the direction $h$. 
Then $\norm{f}_{U^{k+1}}$, the {\em $(k+1)$th Gowers norm of $f$}, is defined as
\[
\norm{f}_{U^{k+1}} ~\defeq~ \abs{\Ex{x,h_1,\ldots h_k \in \F^n}{\Delta_{h_1} \ldots \Delta_{h_k}
    f(x)}}^{1/2^k} \mper
\]
Green and Tao show the following:
\begin{theorem}[\cite{GT09}]\label{thm:gt09}
If a polynomial $P$ of degree $d < |\F|$ satisfies
$\norm{\eF{P}}_{U^{k+1}} \geq \eps$, then there exists a polynomial
$Q$ of degree at  most $k$, such that 
\[
\abs{\ip{\eF{P},\eF{Q}}} ~=~ \abs{\Ex{x \in \F^n}{\eF{P(x)-Q(x)}}} ~\geq~ \eta \mcom
\] 
for some $\eta$ depending only on $\eps$ and $d$. 
\end{theorem}

\newcommand{\gnorm}[2]{\norm{#1}_{U^{#2}}}
\newcommand{\dist}{\mathsf{Dist}}
\newcommand{\rdistance}{1-\frac{1}{|\F|}}

\paragraph{Testing and Decoding Reed-Muller codes beyond the list-decoding radius.}
The above result has a direct interpretation in terms of Reed-Muller codes over $\F$, which raises
the algorithmic question we consider. The Reed-Muller code of order $k$ over $\F^n$ is simply the set
of polynomials of degree at most $k$ over $\F^n$. 
If for a given polynomial $P$ of
degree $d$, there exists a polynomial $Q$ of degree $k$ such that $\Prob{x \in \F^n}{P(x) = Q(x)}
\geq \frac{1}{|\F|} +\eps$, which is the same as saying that $\dist(P,Q) \leq 1 - \frac{1}{|\F|} -
\eps$ (for $\dist(P,Q)$ denoting the normalized Hamming distance),  
then it follows from the definition of Gowers norms and the Gowers Cauchy-Schwarz inequality that $\gnorm{\eF{tP}}{k+1} \geq \eps$ for a nonzero $t\in \F$. 
Then, \cref{thm:gt09} gives that there exists a
polynomial $\tilde{Q}$ of degree $k$ such that $\abs{\ip{\eF{tP},\eF{\tilde{Q}}}} \geq \eta$, which can be
translated to saying that $\dist(P,Q') \leq 1 - \frac{1}{|\F|} - \eta'$ for some $\eta' > 0$ and $Q'$
of degree $k$. 

Thus, the Gowers norm gives an approximate test for checking if for a given $P$,
there exists a $Q$ of degree at most $k$ within Hamming distance $\rdistance - \eps$. If there
exists a $Q$, then the Gowers norm is large and if the Gowers norm is larger than $\eps$, 
then there exists a $Q'$ within distance $\rdistance - \eta'$. This is remarkable because the
list-decoding radius of Reed-Muller of order $k$ codes is only $1-\frac{k}{\abs{\F}}$ for $k < |\F|$
\cite{G10}, and the test works even beyond that. In fact the Hamming distance of a random $P$
is $\rdistance - o(1)$ from all $Q$ of degree $k$ and the test works all the way up to that distance.

However, note that \cref{thm:gt09} only shows that this test works for a $P$ which is a polynomial
of degree $d < |\F|$.
Tao and Ziegler~\cite{TZ} later showed that this test  works even when $P$ is an
  arbitrary function and $k<|\F|$.
For fields of low characteristic, the only general testing results for all functions
of the above flavor (which works beyond the list-decoding radius)
were proved by Samorodnitsky
\cite{Samorodnitsky07} for Reed-Muller codes of order 2 over
$\F_2^n$ and by Green and Tao~\cite{GT05} for Reed-Muller codes of order 2 over $\F_5^n$.

Given the above, it is natural to consider the \emph{decoding}
analogue of the above question: 
\begin{quote}{Given
$P$ of degree $d$ over $\F^n$, if there exists $Q$ of degree $k$ such that $\dist(P,Q) \leq
\rdistance - \eps$, can one find a $Q'$ (in time polynomial in $n$) such that $\dist(P,Q') \leq
\rdistance - \eta$ for some $\eta$ depending on $\eps$?}
\end{quote}
Note that $d,k$ and $\abs{\F}$ are assumed to be constants and
dependence on these is allowed, but not on $n$. Also, observe that there might be
exponentially many such $Q$ since we are in the regime beyond the
list-decoding radius (see \cite{KLP12}), but the question turns out to
be tractable since we only ask for one such $Q$ and allow a loss from
$\eps$ to $\eta$.

Such a decoding question was solved for Reed-Muller codes of order 2 over $\F_2$, 
for any given function $f$ (instead of a polynomial $P$ of bounded degree) by \cite{TW11}.
We solve the above decoding question for polynomials $P$ of degree $d$ and the Reed-Muller code of
order $k$ for $k \leq d < |\F|$. This special case can be interpreted as follows: if $P$ is of
degree-$d$ for some degree $k \leq d < |\F|$, then we can think of $P$ being obtained from some $Q$
of degree $k$ (which is a codeword) by adding the ``noise'' $P-Q$ of degree $d$. Thus, when the
noise is ``structured'' \ie given by a degree-$d$ polynomial, we can decode in the above sense (of
finding \emph{some} codeword within a given distance) even beyond the list-decoding radius. 

Both our algorithmic version of the regularity lemma and the decoding algorithm, are randomized
algorithms that run in time $O(n^d)$ and output the desired objects with high probability. 
This is linear time for regularity since even writing a
polynomial of degree $d$ takes time $\Omega(n^d)$. 
For the question of finding a degree $k$ polynomial within a given distance of $P$, it is possible
that one may be able to do this in time $O(n^k)$, but we do not achieve this.

We remark that we do not give a regularity lemma for the notion of regularity stated by Green and
Tao. We in fact define a related analytical notion  of regularity, and show that one can prove an efficient
regularity lemma for this notion, and also that this notion suffices for their application and our
algorithmic version of it. 
However, our algorithmic version of the regularity lemma (for this notion) 
only works when $|\F| > d$.  This is because our notion of regularity is based on
the Gowers norm, and even to prove the regularity lemma, one needs to use (and modify) the proof of
Green and Tao for the inverse theorem in \cite{GT09}, and their proof only works when $|\F| > d$.

We define other notions later, which work over small fields. However, the above notion based on
Gowers norms is conceptually much simpler and may be useful for other applications. We would like to point out that a similar notion of regularity, called analytic rank, was introduced by Gowers and Wolf~\cite{gowers-wolf-2} and later used by Tao and Ziegler~\cite{TZ12} in their proof of inverse theorem for Gowers norms.

\paragraph{Regularity for Low Characteristics and Efficient Worst-Case to Average-Case Reductions
  for Polynomials.}
Kaufman and Lovett \cite{KL08} develop a more involved notion of regularity to carry over part
of the Green-Tao result to the setting when $|\F|$ is small. 

The technical part of the  Green-Tao proof
proves the following result: If $P$ is a polynomial of degree $d$ such that 
$\bias(P) = \abs{\Ex{x}{\eF{P(x)}}} \geq \delta$, then there exist polynomials $Q_1,\ldots,Q_{M}$ of degree at
most $d-1$, and a function $\Gamma: \F^M \to \F$ such that $P = \Gamma(Q_1,\ldots,Q_M)$. Here $M$
depends only on $d,\delta$ and $|\F|$. This can be read as saying that a biased degree-$d$
polynomial can be computed by few polynomials of degree $d-1$.

Kaufman and Lovett \cite{KL08} manage to derive the above conclusion even when $|\F|$ is
small, by defining and analyzing a more sophisticated notion of regularity. The notion is bit
technical and we defer the description to \cref{sec:notions}. They also use their analog of
the above result to derive an interesting application. Let $P$ be a degree-$d$ polynomial which can
be \emph{weakly approximated} by a lower degree polynomial $Q$, of degree $k < d$. Here by weak
approximation we mean $\ip{\eF{P},\eF{Q}} \geq \delta$. Then Kaufman and Lovett show that $P$ can in
fact be \emph{computed} by lower degree polynomials \ie there exist $Q_1,\ldots,Q_m$ of degree at
most $d-1$ and a function $\Gamma: \F^M \to \F$, such that $P = \Gamma(Q_1,\ldots,Q_M)$. Again, $M$
depends only on $d,\delta$ and $|\F|$.

We consider the question of whether it is possible to \emph{find} the polynomials $Q_1,\ldots,Q_M$
and the function $\Gamma$ efficiently? To this end, we define a version of the Kaufman-Lovett
regularity notion, for which we can prove an algorithmic regularity lemma over low
characteristics. We also use this lemma to give an efficient algorithm for finding the above
polynomials $Q_1,\ldots,Q_M$ with high probability. 

\bigskip

We believe that our notions of regularity and the algorithmic regularity lemmas for these will also
be useful in other applications. We collect the various notions of regularity and provide a more
technical overview in \cref{sec:notions}. We also provide a brief overview of the proofs below.

\subsubsection*{Proof Overview}
The key technical part of both the Green-Tao result and the Kaufman-Lovett result is proving that a
biased polynomial of degree $d$ can be computed by a few polynomials of degree at most $d-1$. Both
proofs proceed by applying a lemma of Bogdanov and Viola \cite{BogV} 
to find a factor $\cF$ of degree $d-1$ such that the polynomials in $\cF$ compute $P$ correctly on \emph{most}
inputs $x \in \F^n$. Then both proofs proceed to refine $\cF$ to $\cF'$ by their respective notions
of regularity, and show that the polynomials in $\cF'$ must in fact compute $P$ exactly. It is an
easy observation, which we make in \cref{sec:bv}, that the Bogdanov-Viola lemma is
algorithmic. 

What remains is then to show that one can refine $\cF$ to $\cF'$ efficiently. As discussed before this is not possible using the original notions of regularity, and we achieve this by
defining related but somewhat different notions of regularity. For these notions, we are able to go
from $\cF$ to $\cF'$ efficiently. But then, we also need to prove that $\cF'$ which is now regular
according to our new notion of regularity, still computes $P$ exactly. For this we need to show that
the Green-Tao (resp. Kaufman-Lovett) proof goes through with our notion. This requires a tighter analysis of both the proofs by Green-Tao and Kaufman-Lovett, and amounts to showing that the new notions of regularity, although are weaker, they still obey similar equidistribution properties that are used in both proofs. 

The key difficulty in proving an algorithmic regularity lemma for polynomials is the same as that in
the case of \Szemeredi's lemma for graphs: proving a certificate of non-regularity. For example, in
the Green-Tao notion, a factor $\cF = \{P_1,\ldots,P_m\}$ is \emph{not regular}, if a linear
combination $\sum_i c_i \cdot P_i$ can be expressed as function $\Gamma(Q_1,\ldots,Q_M)$ for some
lower degree polynomials $Q_1,\ldots,Q_M$ and an appropriate $M$. They then proceed by adding
$Q_1,\ldots,Q_M$ to the factor and removing one of the $P_i$s. However, to provide an algorithmic
version, one needs to \emph{find} $Q_1,\ldots,Q_M$, which is not an easy problem.

We get around this by defining a notion of regularity which requires the Gowers norm of every linear
combination to be small (we call notion $\gamma$-uniformity). 
Now it is easy to check if some linear combination (of say degree $d$) 
has large Gowers norm. However, if this happens, it is not clear how to proceed to refine $\cF$. 
We then look at the
proof of Green and Tao, who show that when a polynomial has large Gowers norm, then an appropriate
derivative of it (when we also think of the direction for the derivative as a variable) 
has large bias. This means that this derivative can now be computed on most inputs
by polynomials of degree $d-1$ (using Bogdanov-Viola), forming a factor $\cF_{d-1}$ 
of degree $d-1$. 
\emph{By induction, we can assume that we can in fact refine $\cF_{d-1}$ to a regular
$\cF_{d-1}'$ factor of degree $d-1$}. If our notion of regularity is good enough to show that a regular
$\cF_{d-1}$ must exactly compute the derivative (and hence also the linear combination), 
we can add all the polynomials in $\cF_{d-1}$ to $\cF$ and proceed. Note that because of the nature
of this induction, the proof of the Green-Tao result and the refinement of a factor to a regular
one, are linked and thus we only obtain an algorithmic lemma with this notion when $|\F| > d$.

For the Kaufman-Lovett version, our notion simply says that the bias of certain linear combinations
of polynomials in the factor and their derivatives is small, 
along with a few other technical
conditions imposed in \cite{KL08}. Although the
notion of regularity here is more complicated, the proof of the algorithmic regularity lemma is in
fact simpler and does not involve an induction as in the above case. However, showing that our
notion suffices for their application requires some work. We defer the details to \cref{sec:notions} and \cref{sec:lowchar}.

\subsection{Definitions}
\begin{dfn}[Factors]
Let $d\geq 0$ and $M_1,\ldots,M_d$ be non-negative integers. By factor of degree $d$ on $\F^n$ we
mean a collection $\cF=(P_{i,j})_{1\leq i\leq d, 1\leq j\leq M_i}$ of polynomials where
$\deg(P_{i,j})= i$ for all $i,j$. 
By the dimension vector of $\cF$ we mean $(M_1,\ldots,M_d)\in
\N^d$ and by the dimension of $\cF$ denoted by $dim(\cF)$ we mean the number of polynomials in the
factor, namely $M_1+\cdots+M_d$. 

Every factor $\cF$ defines a $\sigma$-algebra $\sigma(\cF)$ defined
by atoms of the form $\{x: P_{i,j}(x)=c_{i,j}\}$. We write $\Sigma=\F^{M_1}\times \cdots \times
\F^{M_d}$ and call this the configuration space of $\cF$. Moreover we let $\norm{\cF} \defeq
\abs{\Sigma}$. By abuse of notation we write $\cF:\F^n \rightarrow \Sigma$ for the evaluation map
$\cF(x)= (P_{i,j}(x))_{1\leq i\leq d, 1\leq j \leq M_i}.$

Sometimes if we do not make use of the dimension vector of $\cF$ explicitly, we may write
$\cF=\{P_1,...,P_m\}$. In this notation, we will simply use $i$ as index of the polynomial
$P_i$, and not to denote the degree of $P_i$. 
\end{dfn}

\begin{dfn}[Measurability]
Let $\cF = \{P_1,\ldots,P_m\}$ be a factor of degree $d$ and let $f: \F^n \to \F$ be any given
function. We say that $f$ is measurable in $\cF$ if there exists a function $\Gamma: \F^m \to \F$
such that $f = \Gamma(P_1,\ldots,P_m)$.

We say that $f$ is $\sigma$-close to being measurable if there exists a $\Gamma$  such that
$\Prob{x \in \F^n}{f(x) \neq \Gamma(P_1(x),\ldots,P_m(x))} \leq \sigma$.
\end{dfn}

\begin{dfn}[Refinement]
Let $\cF = \{P_1,\ldots,P_m\}$ be a factor of degree $d$. We say that a factor $\cF' =
\{Q_1,\ldots,Q_M\}$  of degree-$d$ is a refinement of $\cF$ if each polynomial $P_i$ in $\cF$ is
measurable in $\cF'$ \ie there exists $\Gamma: \F^M \to \F^m$, so that 
$(P_1,\ldots,P_m) = \Gamma(Q_1,\ldots,Q_M)$.

We say that $\cF'$ is $\sigma$-close to being a refinement of $\cF$, if there exists $\Gamma: \F^M
\to \F^m$, so that $\Prob{x \in \F^n}{(P_1(x),\ldots,P_m(x)) \neq \Gamma(Q_1(x),\ldots,Q_M(x))} \leq
\sigma$.
\end{dfn}

\begin{dfn}[Derivatives]
For a polynomial $P: \F^n \to \F$ and a point $h \in \F^n$, we define the derivative of $P$ in
direction $h$ as the function $D_h P(x) = P(x+h) - P(x)$. Note that if $P$ is of degree $d$,
then $D_h P$ is of degree $d-1$.

For a function $f: \F^n \to \C$, we will use the derivative to mean the multiplicative derivative
defined as $\Delta_h f(x) = f(x+h) \overline{f(x)}$. Note that if $f = \eF{P}$ for some polynomial
$P$, then $\Delta_h f = \eF{D_h P}$.
\end{dfn}

\section{Approximating a biased polynomial}\label{sec:bv}
The Bogdanov-Viola Lemma~\cite{BogV} states that if a polynomial of 
degree $d$ is biased then it can be approximated by a bounded set of 
polynomials of lower degree. The following is an easy to observe 
algorithmic version of this lemma. We follow the proof of the lemma by Green and Tao
\cite{GT09}.

\begin{lemma}[Algorithmic Bogdanov-Viola lemma]
\label[lemma]{lem:algorithmicBV}
Let $d\geq 0$ be an integer, and $\delta, \sigma, \beta \in (0,1]$ be parameters. 
There exists a randomized algorithm, that given query access to a 
polynomial $P:\F^n\rightarrow \F$ of degree $d$ with
$$
\bias(P) \geq \delta,
$$
runs in time $\rO_{\delta, \beta, \sigma}(n^d)$, and with probability 
$1-\beta$ returns functions $\tP:\F^n \rightarrow \F$ and 
$\Gamma:\F^C\rightarrow \F$, and a set of polynomials $P_1,...,P_C$, 
where $C \leq \frac{|\F|^5}{\delta^2 \sigma \beta}$ and $\deg(P_i)<d$ for all $i\in [C]$, for which 

\begin{itemize}
\item $\Pr_x(P(x) \neq \tP(x))\leq \sigma$, and
\item $\tP(x)= \Gamma(P_1(x),...,P_C(x))$.
\end{itemize}

\end{lemma}
\begin{proof}
The proof will be an adaptation of the proof from~\cite{GT09}. Given query 
access to the polynomial $P$, we can compute the explicit description of 
$P$ in $O(n^d)$ queries. For every $a\in \F$ define the measure 
$\mu_a(t) \defeq \Pr(P(x)=a+t)$. It is easy to see that if 
$\bias(P(x))\geq \delta$ then, for every $a\neq b$, 
\begin{equation}\label{eq:mufar}
\norm{\mu_a - \mu_b}\geq \frac{4\delta}{|\F|}.
\end{equation}

We will try to estimate each of these distributions. Let 
$$
\tmu_a(t)\defeq \frac{1}{C} \sum_{1\leq i \leq C} \mathbbm{1}_{P(x_i)= a+t},
$$
where $C> \frac{|\F|^5}{\delta\cdot \beta_1}$, and $x_1,x_2,...,x_C\in \F^n$ 
are chosen uniformly at random. Therefore by an application of Chebyshev's inequality
$$
\Pr\left(\left|\tmu_a(t)-\mu_a(t)\right|> \frac{\delta}{2|\F|^2}\right) < \frac{\beta_1}{|\F|}, 
$$
for all $t\in \F$ and therefore
\begin{equation}\label{eq:mu}
\Pr\left( \norm{\tmu_a - \mu_a }> \frac{\delta}{2|\F|^2}\right) < \beta_1.
\end{equation}

Now we will focus on approximating $P(x)$. Remember that 
$D_hP(x) = P(x+h)- P(x)$ is the additive derivative of $P(x)$ in direction $h$. We have
$$
\Pr_h(D_hP(x)= r) = \Pr_h(P(x+h)-P(x)=r) = \mu_{P(x)}(r),
$$
where $h\in \F^n$ is chosen uniformly at random. Let $\mathbf{h}= (h_1,...,h_C)\in (\F^n)^C$ 
be chosen uniformly at random, where $C$ is a sufficiently large constant to be chosen later. Define
$$
\mu_\obs^{(x)}(t) \defeq \frac{1}{C} \sum_{1\leq i\leq C} \mathbbm{1}_{D_{h_j}P(x)=t}, 
$$
and let
$$
\tP_h(x) \defeq \argmin_{r\in \F} \norm{\tmu_r - \mu_\obs^{(x)}}.
$$
Now choosing $C\geq \frac{|\F|^5}{\delta^2 \sigma \beta_2}$ follows
\begin{equation}
\Pr_h (\tP_h(x) \neq P(x)) \leq \Pr_h \left( \norm{\mu_\obs^{(x)} - \mu_{P(x)}} \geq \frac{\delta}{|\F|}\right) \leq \sigma \beta_2,
\end{equation}
where the first inequality follows from (\ref{eq:mufar}) and (\ref{eq:mu}). Therefore
$$
\Pr_{x,h} (\tP_h(x) \neq P(x) )= \E_x \E_h \mathbbm{1}_{\tP_h(x) \neq P(x)}\leq \sigma \beta_2,
$$
and thus
$$
\Pr_h\left[ \Pr_x\left(\tP_h(x)\neq P(x)\right) \geq \sigma \right]\leq \beta_2.
$$

Let $P_i\defeq D_{h_i}P$, so that $P_i$ is of degree $\leq d$ and 
$\tP_h$ is a function of $P_1,...,P_C$. Now setting 
$\beta_1:= \frac{\beta}{2|\F|^2}$ and $\beta_2:=\frac{\beta}{2}$ finishes the proof.  
\end{proof}

\section{Notions of Regularity}\label{sec:notions}
Notions of regularity similar to the following one defined by Green and Tao \cite{GT09}, 
play an important role in concepts related to polynomials. 

\begin{dfn}[$F$-regularity]\label[dfn]{dfn:regularfactor}
Let $\cF$ be a factor of degree $d$, and let $F:\N\rightarrow \N$ be a 
growth function. We say that $\cF$ is $F$-regular if for every collection 
of coefficients $\{c_{i,j}\in \F\}_{1\leq i\leq d, 1\leq j \leq M_i}$, we have
$$
\rank_{k-1}\left( \sum_{i=1}^d \sum_{j=1}^{M_i} c_{i,j} P_{i,j} \right) \geq F(\dim(\cF)),
$$
where $k$ is the largest $i$ such that there is a nonzero $c_{i,j}$. 
\end{dfn}

The following, by now well-known, lemma states that given a polynomial 
factor one can refine it to a regular one.

\begin{lemma}[Regularity Lemma for Polynomials] \label[lemma]{lem:gtregularity}
Let $d\geq 1$, $F:\N\rightarrow \N$ be a growth function, and let $\cF$ 
be a factor of degree $d$. Then there exists an $F$-regular extension 
$\cF'$ of $\cF$ of same degree $d$, satisfying 
$$
\dim(\cF') =\rO_{F, d, \dim(\cF)}(1). 
$$
\end{lemma}
We provide a proof for two reasons, self-containment and the fact that 
we will use the idea of this proof later in the section. 
\begin{proof}
We shall induct on the dimension vectors $(M_1,...,M_d)$. Notice that the 
dimension vectors of a factor of degree $d$ takes values in $\N^d$, 
and we will use the reverse lexicographical ordering on this space for 
our induction. The proof for $d=1$ is obvious, because non-zero linear 
functions have infinite $\rank_0$, unless there is a linear dependency 
between the polynomials in $\cF$, which we can simply discard by removing 
the polynomials that can be written as a linear combination of the others. 

If $\cF$ is already $F$-regular, then we are done. Otherwise, there is a set of coefficients $\{c_{i,j}\}_{1\leq i\leq d, 1\leq j\leq M_i}$, not all zero, such that
$$
\rank_{k-1}\left( \sum_{i=1}^d \sum_{j=1}^{M_i} c_{i,j} P_{i,j} \right) < F(\dim(\cF)).
$$
Without loss of generality assume that $c_{k,M_k}\neq 0$. The above 
inequality means that $P_{k,M_k}$ can be written as a function of 
the rest of the polynomials in the factor and a set of at most 
$F(\dim(\cF))$ polynomials of degree at most $k-1$. We will replace 
$P_{k,M_k}$ with these new polynomials. The new factor will have the 
following dimension vector
$$
(M_1,\ldots, M_{k-1}+ F(\dim(\cF)), M_k-1, M_{k+1}, \ldots, M_d).
$$
Now by the induction hypothesis the above can be regularized. 
\end{proof}

The nature of the proof in regularity lemmas of the above type 
results in Ackermann like functions even for ``reasonable'' growth functions $F$. 

\subsection{Unbiased Factors}
In what follows we suggest a new type of regularity which seems better 
suited for algorithmic reasons. The definition is inspired by the fact 
that in the applications of the regularity lemma one is usually interested 
in the linear combinations of the polynomials in the factor being unbiased. 
This definition will not suffice for our applications, but serves as a good 
starting point for stronger notions that we will introduce later. 

\begin{dfn}[$\gamma$-unbiased factors]
\label[dfn]{dfn:unbiasedfactor}
Let $\cF$ be a factor of degree $d$, and let $\gamma:\N \rightarrow \R^+$ 
be a decreasing function. We say that $\cF$ is $\gamma$-unbiased if 
for every collection of coefficients $\{c_{i,j}\in \F\}_{1\leq i\leq d, 1\leq j \leq M_i}$, we have
$$
\bias(\sum_{i,j} c_{i,j} P_{i,j}) \leq \gamma(\dim(\cF)).
$$
\end{dfn}

The following can be thought of as an algorithmic analogue of the regularity lemma. 

\begin{lemma}[Unbiased almost Refinement]
\label[lemma]{lem:unbiasedrefinement}
Let $d\geq 1$ be an integer. Suppose $\gamma:\N \rightarrow \R^+$ is 
a decreasing function and $\sigma, \rho \in (0,1]$. There is a 
randomized algorithm that given a factor $\cF$ of degree $d$, runs 
in time $\rO_{\gamma,\rho,\sigma, \dim(\cF)}(n^d)$ and  with probability 
$1-\rho$ returns a $\gamma$-unbiased factor $\cF'$ with 
$\dim(\cF')=\rO_{\gamma,\rho,\sigma, \dim(\cF)} (1),$ such 
that $\cF'$ is $\sigma$-close to being a refinement of $\cF$.
\end{lemma}
\begin{proof}
The proof idea is similar to that of \cref{lem:gtregularity} 
in the sense that we do the same type of induction. The difference 
is that at each step we will have to control the errors that we 
introduce and the probability of correctness. At all steps in the 
proof without loss of generality we will assume that the polynomials 
in the factor are linearly independent, because otherwise we can always 
detect such a linear combination in $\rO_{\gamma,\rho,\sigma, \dim(\cF)}(n^d)$ 
time and remove a polynomial that can be written as a linear 
combination of the rest of the polynomials in the factor.

The base case for $d=1$ is simple, a linearly independent set of 
non-constant linear polynomials is not biased at all, namely it 
is $0$-unbiased. Assume that $\cF$ is $\gamma$-biased, then there 
exists a set of coefficients $\{c_{i,j}\in \F\}_{1\leq i\leq d, 1\leq j \leq M_i}$ such that
$$
\bias(\sum_{i,j} c_{i,j} P_{i,j}) \geq \gamma(\dim(\cF)). 
$$
To detect this, we will use the following algorithm: 

We will estimate bias of each of the $|\F|^{\dim(\cF)}$ linear 
combinations and check whether it is greater than $\frac{3\gamma(\dim(\cF))}{4}$. 
To do so, for each linear combination $\sum_{i,j} c_{i,j} P_{i,j}$ independently
select a set of vectors $x_1,...,x_C$ uniformly at random from $\F^n$, 
and let $\widetilde{\bias}(\sum_{i,j} c_{i,j} P_{i,j})\defeq \left|\frac{1}{C} \sum_{\ell\in [C]} e_\F(y_\ell)\right|$, 
where $y_\ell=\sum_{i,j}c_{i,j}\cdot P_{i,j}(x_\ell)$. Choosing 
$C=\rO_{\dim(\cF)}\left(\frac{1}{\gamma(\dim(\cF))^2}\log(\frac{1}{\rho})\right)$, 
we can distinguish bias $\geq \gamma$ from bias $\leq \frac{\gamma}{2}$, 
with probability $1-\rho'$, where $\rho':= \frac{\rho}{4|\F|^{\dim(\cF)}}$. 
Let $\sum_{i,j} c_{i,j} P_{i,j}$ be such that the estimated bias was above 
$\frac{3\gamma(\dim(\cF))}{4}$ and $k$ be its degree. We will stop if there 
is no such linear combination or if the factor is of degree $1$. Since by 
a union bound with probability at least 
$1-\frac{\rho}{4}$, $\bias(\sum_{i,j} c_{i,j} P_{i,j}) \geq \frac{\gamma(\dim(\cF))}{2}$, 
by \cref{lem:algorithmicBV} we can find, with probability $1-\frac{\rho}{4}$, 
a set of polynomials $Q_1,...,Q_r$ of degree $k-1$ such that
\begin{itemize}
\item $\sum_{i,j} c_{i,j} P_{i,j}$ is $\frac{\sigma}{2}$-close to a function of $Q_1,...,Q_r$,
\item $r\leq \frac{16|\F|^5}{\gamma(\dim(\cF))^2 \cdot \sigma \cdot \rho}$.
\end{itemize}
We replace one polynomial of highest degree that appears in 
$\sum_{i,j} c_{i,j} P_{i,j}$ with polynomials $Q_1,...,Q_r$. 

We will prove by the induction that our algorithm satisfies the statement 
of the lemma. For the base case, if $\cF$ is of degree $1$, our 
algorithm does not refine $\cF$ by design. Notice that since 
we have removed all linear dependencies, $\cF$ is in fact $0$-unbiased in this case. 

Now given a factor $\cF$, if $\cF$ is $\gamma$-biased, then with probability 
$1-\rho'$ our algorithm will refine $\cF$. With probability $1-\frac{\rho}{4}$ 
the linear combination used for the refinement is 
$\frac{\gamma(\dim(\cF))}{2}$-biased. Let $\widetilde{\cF}$ be the outcome 
of one step of our algorithm.  With probability $1-\frac{\rho}{4}$, 
$\widetilde{\cF}$ is $\frac{\sigma}{2}$-close to being a refinement of 
$\cF$. Using the induction hypothesis with parameters 
$\gamma, \frac{\sigma}{2}, \frac{\rho}{4}$ we can find, with probability 
$1-\frac{\rho}{4}$, a $\gamma$-unbiased factor $\cF'$ which is $\frac{\sigma}{2}$-close 
to being a refinement of $\widetilde{\cF}$ and therefore, with probability at least 
$1-(\frac{\rho}{4}+\frac{\rho}{4}+\frac{\rho}{4}+\rho')> 1-\rho$, is 
$\sigma$-close to being a refinement of $\cF$. 
\end{proof}

One of the main reasons one is interested in regular factors is that 
they partition the space into almost equal atoms (See for example 
\cite{KL08, TZ12, GT09}). We observe that the same is true for unbiased factors. 

\begin{lemma}[Equidistribution for unbiased factors]\label[lemma]{lem:equidistribution}
Suppose that $\gamma:\N\rightarrow \R^+$ is a decreasing function. 
Let $\cF=\{P_1,\ldots, P_m\}$ be a $\gamma$-unbiased factor of 
degree $d>0$. Suppose that $b\in \F^m$. Then
$$
\Pr_{x\in \F^n}[\cF(x)=b] = \frac{1}{\norm{\cF}}\pm \gamma(\dim(\cF)),
$$
where $\norm{\cF}=|\Sigma|$ is the number of atoms of $\cF$. 
\end{lemma}
\begin{proof}
\begin{align*}
\Pr_{x\in \F^n}[\cF(x)=b]&= \E_x\left[ \prod_{i\in [m]} \frac{1}{|\F|} \sum_{\lambda_i=0}^{|\F|-1} e_\F(\lambda_i(P_i(x)-b_i)) \right] \\ &=
\frac{1}{|\F|^m} \sum_{\lambda\in \F^m} \E_{x\in \F^n}\left[ e_\F \left( \sum_{i\in [m]} \lambda_i\left(P_i(x)-b_i\right) \right) \right]\\ &=
\frac{1}{|\F|^m} \left( 1\pm \gamma(\dim(\cF)) |\F|^m \right) = \frac{1}{\norm{\cF}} \pm \gamma(\dim(\cF)),
\end{align*}
where the last inequality follows from the fact that for every 
$\lambda\neq \underline{0}$, the polynomial 
$\sum_i \lambda_i\left(P_i(x)-b_i\right)$ is nonzero and therefore 
has bias less than $\gamma(\dim(\cF))$. 
\end{proof}

Let $P:\F^n \rightarrow \F$ be a function, and let $\cF$ be a polynomial factor 
such that $P$ is measurable in $\cF$. This means that there exists a function 
$\Gamma:\F^{\dim(\cF)}\rightarrow \F$ such that $P(x)= \Gamma(\cF(x))$. 
Although we know that such a $\Gamma$ exists, but we do not have explicit 
description of $\Gamma$. It is a simple but useful observation that we can provide 
query access to $\Gamma$ in case $\cF$ is unbiased. It is worth recording this as 
the following lemma. 

\begin{lemma}[Query access]\label[lemma]{lem:query}
Suppose that $\beta\in (0,1]$ and let $\cF=\{P_1,\ldots, P_m\}$ be 
a $\gamma$-unbiased factor of degree $d$, where $\gamma:\F^n\rightarrow \F$ 
is a decreasing function which decreases suitably fast. Suppose that we 
are given query access to a function $f:\F^n\rightarrow \F$, where $f$ 
is measurable in $\cF$, namely there exists $\Gamma:\F^{\dim(\cF)} \rightarrow \F$ 
such that $f(x)= \Gamma(\cF(x))$. There is a randomized algorithm that, takes 
as input a value $a=(a_1,\ldots,a_m)\in \F^m$, makes a single query to $f$, 
runs in $O(n^d)$, and returns a value $y\in \F$ such that $\Gamma(a)=y$ with 
probability greater than $1-\beta$.
\end{lemma}
\begin{proof}
Choose $\gamma(x):= \frac{1}{2p^{x}}$. By \cref{lem:equidistribution} 
$$
|\cF^{-1}(a)|\geq \frac{1}{\norm{\cF}} - \frac{1}{2\norm{\cF}} = \frac{1}{2\norm{\cF}}.
$$
Let $x_1,...,x_K$ be chosen uniformly at random from $\F^n$, where 
$K\geq 2\norm{\cF} \log \frac{1}{\beta}$. Then with probability at least $1-\beta$, 
there is $i^*\in [K]$ such that $\cF(x_{i^*})=(a_1,\ldots,a_m)$. 
This means that $x_{i^*}\in \cF^{-1}(a)$. Notice that we can 
find this $i^*$, since we have explicit description of polynomials 
in $\cF$ and we can evaluate them on each $x_i$. Our algorithm will 
query $f$ on input $x_{i^*}$ and return $f(x_{i^*})$. 
\end{proof}

\subsection{Uniform Factors}
As it turns out, one sometimes needs stronger equidistribution properties from a factor. 
For example it is sometimes important to have an almost independence property, 
similar to \cref{lem:equidistribution}, 
for the factor on a random $k$-dimensional parallelepiped. Although this holds for 
a regular factor (See \cite{GT09} Proposition~4.5), our notion of unbiased factors 
fails to satisfy such properties. To this end we define the following notion of a 
uniform factor.

\begin{dfn}[$\gamma$-uniform factors]
\label[dfn]{dfn:uniformfactor}
Let $\cF$ be a factor of degree $d$, and let $\gamma:\N \rightarrow \R^+$ 
be a decreasing function. We say that $\cF$ is $\gamma$-uniform if for 
every collection of coefficients $\{c_{i,j}\in \F\}_{1\leq i\leq d, 1\leq j \leq M_i}$ 
with $\deg\left(\sum_{i,j} c_{i,j} P_{i,j} \right)=k$,
$$
\norm{e_\F\bigg(\sum_{i,j} c_{i,j} P_{i,j}\bigg)}_{U^k} \leq \gamma(\dim(\cF)).
$$
\end{dfn} 

\begin{remark}
Notice that every $\gamma$-uniform factor is also $\gamma$-unbiased since 
$$
|\E[f(x)]| \leq \norm{f}_{U^k},
$$ 
and indeed uniformity is a stronger notion of regularity. 
\end{remark}

In \cref{sec:uniformrefinement} we will present an algorithm 
to refine a given factor to a uniform one. We will need an assumption 
that $|\F|$ is larger than the degree of the factor. The reason is that 
our proof uses division by $d!$ where $d$ is the degree of the 
factor, and therefore we will need $|\F|$ to be greater than $d$.

\subsection{Strongly Unbiased Factors}
As mentioned in previous section, notions of regularity 
(\cref{dfn:regularfactor}) and uniformity (\cref{dfn:uniformfactor}) fail to address the case when the field $\F$ has small characteristic. 
 Kaufman and Lovett~\cite{KL08} introduce a stronger notion of regularity 
 to prove the ``bias implies low rank'' in the general case. To define 
 their notion of regularity we first have to define the derivative 
 space of a factor. During this section, for simplicity in the writing, 
 we will denote a factor $\cF$ by a set of its polynomials $\{P_1,...,P_m\}$ 
 where $m=\dim(\cF)$. In this notation, we no longer use the index $i$ to denote the degree of $P_i$. 

\begin{definition}[Derivative Space]
Let $\cF=\{P_1,...,P_m\}$ be a polynomial factor over $\F^n$. Recall that
$$
D_hP(x) ~=~ P(x+h)-P(x),
$$
is the additive derivative of $P$ in direction $h$, where $h$ is a vector from $\F^n$. We define
$$
\Der(\cF) ~\defeq~ \{ (D_hP_i)(x): i\in [m], h\in \F^n\},
$$
as the derivative space of $\cF$. 
\end{definition}

\begin{definition}[Strong regularity of polynomials~\cite{KL08}]
\label[definition]{dfn:strongregularity}
Suppose that $F:\N\rightarrow \N$ is an increasing function. 
Let $\cF=\{P_1,\ldots,P_m\}$ be a polynomial factor of degree $d$ and $\Delta:\cF\rightarrow \N$ 
assigns a natural number to each polynomial in $\cF$. We say that 
$\cF$ is strongly $F$-regular with respect to $\Delta$ if
\begin{enumerate}
\item For every $i\in [m]$, $1\leq \Delta(P_i) \leq \deg(P_i)$. 
\item For any $i\in [m]$ and $r> \Delta(P_i)$, there exist a function $\Gamma_{i,r}$ such that
$$
P_i(x+y_{[r]})= \Gamma_{i,r}\big( P_j(x+ y_J): j\in [m], J\subseteq [r], |J|\leq \Delta(P_j) \big),
$$
where $x,y_1,...,y_r$ are variables in $\F^n$, and $y_J\defeq \sum_{k \in J} y_k$ for every $J\subseteq [r]$,
\item For any $r\geq 0$, let $\cC=\{c_{i,I}\in \F\}_{i\in [m], I\subseteq [r], |I|\leq \Delta(P_i)}$ be a collection of field elements. Define
$$
Q_\cC(x,y_1,...,y_r)\defeq \sum_{i\in [m], I\subseteq [r], |I|\leq \Delta(P_i)} c_{i,I}\cdot P_i(X+Y_I).
$$
Let $\cF'\subseteq \cF$ be the set of all $P_i$'s which appear in $Q_\cC$. 
There \underline{do not} exist polynomials $\tP_1,...,\tP_\ell\in \Der(\cF')$, 
and $I_1,...,I_\ell\subseteq [r]$, with $l\leq F(\dim(\cF))$, for which $Q_\cC$ can be expressed as 
$$
\Gamma\left(\tP_1(x+y_{I_1}),...,\tP_\ell(x+y_{I_\ell})\right),
$$
for some function $\Gamma:\F^\ell \rightarrow \F$.
\end{enumerate}
\end{definition}

The following lemma states that every polynomial factor can be refined to a strongly regular factor. 

\begin{lemma}[Strong-Regularity Lemma, \cite{KL08}]\label[lemma]{lem:strongregularity}
Let $F:\N\rightarrow \N$ be an increasing function. Let $\cF=\{P_1,...,P_m\}$ be a 
polynomial factor of degree $d$. There exists a refinement $\cF'=\{P_1,\ldots, P_r\}$ 
of same degree $d$ along with a function $\Delta:\cF\rightarrow \N$ such that
\begin{itemize}
\item $\cF'$ is strongly $F$-regular with respect to $\Delta$,
\item $\dim(\cF')= \rO_{F, \dim(\cF), d}(1)$. 
\end{itemize}
\end{lemma}

Moreover they prove that if a polynomial is approximated by a strongly 
regular factor, then it is in fact measurable with respect to the factor. 

\begin{lemma}[\cite{KL08}]\label[lemma]{lem:approxtocomputelowchar}
For every $d$ there exists a constant $\epsilon_d=2^{-\Omega(d)}$ such 
that the following holds. Let $P:\F^n\rightarrow \F$ be a polynomial of 
degree $d$, $P_1,...,P_s$ be a strongly regular collection of polynomials of degree $d-1$ with degree bound $\Delta$ and $\Gamma:\F^s \rightarrow \F$ 
be a function such that $\Gamma(P_1,...,P_s)$ is $\epsilon_d$-close to $P$. Then there exists $\Gamma':\F^s\rightarrow \F$ such that
$$
P(x)= \Gamma'(P_1(x),...,P_s(x)).
$$
\end{lemma}

\cref{lem:strongregularity} and \cref{lem:approxtocomputelowchar} 
combined with Bogdanov-Viola Lemma immediately gives the following theorem.

\begin{theorem}[Bias implies low rank for general fields~\cite{KL08}]
\label{thm:biasimplieslowrank}
Let $P:\F^n\rightarrow \F$ be a polynomial of degree $d$ such that $\bias(P)\geq \delta>0$. Then $\rank_{d-1} \leq c(d,\delta, \F)$. 
\end{theorem}

In \cref{sec:lowchar} we will give an algorithmic version of this 
theorem. To do this, as in previous sections, we will have to use a notion 
of regularity which is well suited for algorithmic applications. Uniform 
factors (\cref{dfn:uniformfactor}) fail to address fields of low 
characteristics, for the reason that to refine a factor to a uniform factor 
we make use of division by $d!$ which is not possible in fields with $|\F|\leq d$. 
Strong regularity of \cite{KL08} suggests how to extend the notion of unbiased 
factors to a stronger version in quite a straight forward fashion. 

\begin{definition}[Strongly $\gamma$-unbiased factors]
\label[definition]{dfn:stronglyunbiased}
Suppose that $\gamma:\N\rightarrow \R^+$ is a decreasing function. 
Let $\cF=\{P_1,...,P_m\}$ be a polynomial factor of degree $d$ over $\F^n$ and let $\Delta:\cF\rightarrow \N$ 
assign a natural number to each polynomial in the factor. We say that $\cF$ is 
strongly $\gamma$-unbiased with degree bound $\Delta$ if 
\begin{enumerate}
\item For every $i\in [m]$, $1\leq \Delta(P_i)\leq \deg(P_i)$. 
\item For every $i\in [m]$ and $r>\Delta(P_i)$, there exists a function $\Gamma_{i,r}$ such that
$$
P_i(x+y_{[r]})= \Gamma_{i,r}\left(P_j(x+y_J): j\in [m], J\subseteq [r], |J|\leq \Delta(P_j)\right).
$$
\item For any $r\leq 0$ and not all zero collection of field elements 
$\cC=\{c_{i,I}\in \F\}_{i\in [m], I\subseteq [r], |I|\leq \Delta(P_i)}$ we have
$$
\bias\left( \sum_{i\in [m], I\subseteq [r], |I|\leq \Delta(P_i)} c_{i,I}\cdot P_i(x+y_I)\right) \leq \gamma(\dim(\cF)),
$$
for random variables $x,y_1,\ldots,y_r\in \F^n$.
\end{enumerate}
We will say $\cF$ is $\Delta$-bounded if it only satisfies (1) and (2).
\end{definition}

In \cref{sec:lowchar} we will show how to refine a given factor to a 
strongly unbiased one. This will allow us prove an algorithmic analogue 
of \cref{lem:approxtocomputelowchar}, which will be presented in \cref{sec:approxtocomputelowchar}.

\section{Uniform Refinement}\label{sec:uniformrefinement}
In this section we shall assume that the field $\F$ has large characteristic, 
namely $|\F|$ is greater than the degree of the polynomials and the degree of 
the factors that we study. We will prove the following lemma which states that 
we can efficiently refine a given factor to a uniform factor. 

\begin{lemma}[Uniform Refinement]
\label[lemma]{lem:uniformrefinement}
Suppose $d<|\F|$ is a positive integer and $\rho\in (0,1]$ is a parameter. There 
is a randomized algorithm that, takes as input a factor $\cF$ of degree $d$ over 
$\F^n$, and a decreasing function $\gamma:\N \rightarrow \R^+$, runs in time 
$\rO_{\rho, \gamma, \dim(\cF)}(n^d)$, and with probability $1-\rho$ outputs a $\gamma$-uniform 
factor $\widetilde{\cF}$, where $\widetilde{\cF}$ is a refinement of $\cF$, is of 
same degree $d$, and $|\dim(\widetilde{\cF})|\ll_{\sigma, \gamma, \dim(\cF)} 1$.
\end{lemma}

\begin{note}
Notice that unlike \cref{lem:unbiasedrefinement}, here $\cF'$ is an exact 
refinement of $\cF$. This will be achieved by the use of \cref{prop:approxtocompute_new} 
which roughly states that approximation of a polynomial by a uniform factor implies its exact computation.
\end{note}

The following proposition of Green and Tao states that if a polynomial is approximated well 
enough by a regular factor, then it is computed by the factor. 

\begin{prop}[\cite{GT09}]
\label[prop]{prop:approxtocompute}
Suppose that $d\geq 1$ is an integer. There exists a constant $\sigma_d>0$ 
such that the following holds. Suppose that $P:\F^n \rightarrow \F$ 
is a polynomial of degree $d$ and that $\cF$ is an $F$-regular factor 
of degree $d-1$ for some increasing function $F:\N\rightarrow \N$ which increases suitably 
rapidly in terms of $d$. Suppose that $\tP:\F^n \rightarrow \F$ is an 
$\cF$-measurable function and that $\Pr(P(x)\neq \tP(x)) \leq \sigma_d$. 
Then $P$ is itself $\cF$-measurable.
\end{prop}

Our proof of \cref{lem:uniformrefinement} will require the following 
variant of \cref{prop:approxtocompute}, the proof of which will be 
deferred to Sections \ref{sec:parallelequidistribution} and \ref{sec:approxtocompute_new}. 
This proposition shows that the uniformity notion is strong enough to imply results from \cite{GT09} 
such as \cref{prop:approxtocompute}.

\begin{prop}\label[prop]{prop:approxtocompute_new}
Suppose that $d\geq 1$ is an integer. There exists 
$\sigma_{\ref{prop:approxtocompute_new}}(d)>0$ such that the following holds. 
Suppose that $P:\F^n \rightarrow \F$ is a polynomial of degree $d$ and that 
$\cF$ is a $\gamma_{\ref{prop:approxtocompute_new}}$-uniform factor of degree 
$d-1$ for some decreasing function $\gamma_{\ref{prop:approxtocompute_new}}$ which 
decreases suitably rapidly in terms of $d$. Suppose that $\tP:\F^n \rightarrow \F$ 
is an $\cF$-measurable function and that $\Pr(P(x)\neq \tP(x)) \leq \sigma_{\ref{prop:approxtocompute_new}}(d)$. 
Then $P$ is itself $\cF$-measurable. 
\end{prop}

\cref{lem:uniformrefinement} and \cref{prop:approxtocompute_new} 
are somewhat entangled: We will prove \cref{lem:uniformrefinement} by induction 
on the dimension vector of $\cF$. For the induction step, we check whether 
there is a linear combination of polynomials in $\cF$ that has large Gowers norm. 
We will use this to replace a polynomial $P$ from $\cF$ with a set of 
lower degree polynomials. To do this, we first approximate $P$ with a few lower 
degree polynomials $Q_1,...,Q_r$, then use the induction hypothesis to refine 
$\{Q_1,...,Q_r\}$ to a uniform factor $\{\tQ_1,...,\tQ_{r'}\}$ and use 
\cref{prop:approxtocompute_new} to conclude that $P$ is measurable 
in $\{\tQ_1,...,\tQ_{r'}\}$. 

\begin{proofof}{of \cref{lem:uniformrefinement}}
Let $\cF$ be a factor of degree $d$. Similar to the proof of 
\cref{lem:unbiasedrefinement} we will induct on dimension vector 
$(M_1,...,M_d)$. At all steps of the proof we may assume that all polynomials 
in $\cF$ are homogeneous, that is, all multinomials are of same degree. This 
is because otherwise, we may replace each polynomial $P_{i,j}$ with $i$ 
homogeneous polynomials that sum up to $P_{i,j}$. Moreover, as in the proof 
of \cref{lem:unbiasedrefinement}, without loss of generality we may 
assume that the polynomials in $\cF$ are linearly independent, since otherwise 
we can detect and remove linear dependencies in $\rO_{\dim(\cF)}(n^d)$ time. 

Given the above assumptions, the base case for $d=1$ is simple, a 
linearly independent homogeneous factor is $0$-uniform. Let $d>1$ and 
assume that $\cF$ is not $\gamma$-uniform, then there exists a set of 
coefficients $\{c_{i,j}\in \F\}_{1\leq i\leq d, 1\leq j\leq M_i}$ such that 
$$
\norm{e_\F\big(\sum_{i,j}c_{i,j}\cdot P_{i,j}\big)}_{U^k} \geq \gamma(\dim(\cF)),
$$ 
where $k=\deg(\sum_{i,j}c_{i,j}\cdot P_{i,j})$. We will use the following 
randomized algorithm to detect such a linear combination:

We will estimate $\norm{e_\F(\sum_{i,j}c_{i,j}\cdot P_{i,j})}_{U^k}^{2^k}$ 
for each of the $|\F|^{\dim(\cF)}$ linear combinations and check 
whether it is greater than $\frac{3\gamma(\dim(\cF))^{2^k}}{4}$. 
To do this, as in the proof of \cref{lem:unbiasedrefinement}, independently 
for each linear combination $\sum_{i,j}c_{i,j}\cdot P_{i,j}$ select 
$C$ sets of vectors $x^{(\ell)},y^{(\ell)}_1,...,y^{(\ell)}_k$ chosen uniformly 
at random from $\F^n$ for $1\leq \ell\leq C$, 
and compute 
$$
\left| \sum_{\ell \in [C]} e_\F\bigg( \sum_{I\subseteq [k]} (-1)^{|I|}\cdot \sum_{i,j}c_{i,j}\cdot P_{i,j}(x^{(\ell)}+ y^{(\ell)}_I) \bigg)\right|
$$

Choosing $C=O_{\dim(\cF)}\left(\frac{1}{\gamma(\dim(\cF))^2}\log(\frac{1}{\rho})\right)$, 
we can distinguish between Gowers norm $\leq \frac{\gamma(\cF)}{2^{1/2^k}}$ and 
Gowers norm $\geq \gamma(\cF)$, with $\rho'=\frac{\rho}{4|\F|^{\dim(\cF)}}$ probability 
of error.

Let $Q\defeq \sum_{i,j} c_{i,j}\cdot P_{i,j}$, with $\deg(Q)=k$, 
the detected linear combination for which the estimated 
$\norm{e_\F(Q)}_{U^k}^{2^k}$ is greater than $\frac{3\gamma(\dim(\cF))^{2^k}}{4}$. 
By a union bound on all the linear combinations, with probability at 
least $1-\frac{\rho}{4}$, $\norm{e_\F(Q)}_{U^k}^{2^k}\geq \frac{\gamma(\dim(\cF))^{2^k}}{2}$. Notice that 
$$
\bias(DQ) = \norm{e(Q)}_{U^k}^{2^k} \ge \frac{\gamma\left(\dim(\cF)\right)}{2},
$$
where $DQ:(\F^n)^k\rightarrow \F$ is defined as
$$
DQ(y_1,...,y_k)\defeq D_{y_1}D_{y_2}\cdots D_{y_k}Q(x).
$$
Let $\tsigma$ and $\tgamma$ be as in \cref{prop:approxtocompute_new}. By \cref{lem:algorithmicBV} we can find, with probability 
$1-\frac{\rho}{4}$, a set of polynomials $Q_1,...,Q_r$ of degree $k-1$ such that 
\begin{itemize}
\item $DQ$ is $\tsigma$-close to being measurable with respect to $Q_1,...,Q_r$. 
\item $r\leq \frac{8|\F|^5}{\gamma(\dim(\cF))^{2^{k+1}}\tsigma \rho}$.
\end{itemize}

By the induction hypothesis, with probability $1-\frac{\rho}{4}$, 
we can refine $\{Q_1,...,Q_r\}$ to a new factor $\{\tQ_1, ..., \tQ_{r'}\}$, 
which is $\tgamma$-uniform. It follows 
from \cref{prop:approxtocompute_new} that $DQ$ is in fact 
measurable in $\{\tQ_1, ..., \tQ_{r'}\}$. Since $|\F|>k$ we can 
write the following Taylor expansion
\begin{equation}\label{eq:divisionbykfact}
Q(x)= \frac{DQ(x,...,x)}{k!}+ R(x),
\end{equation}
where $R(x)$ is a polynomial of degree $\leq k-1$ which can be 
computed from $Q$. Thus we can replace a maximum degree polynomial 
$P_{i,j}$, that appears in $Q$, with polynomials $\tQ_1, ..., \tQ_{r'}$, 
and $R(x)$. The proof follows by using the induction hypothesis for 
parameters $\frac{\rho}{4}$ and $\gamma$.
\end{proofof}

\subsection{Equidistribution over parallelepipeds}
\label{sec:parallelequidistribution}
Since $\gamma$-uniform factors are also $\gamma$-unbiased, they 
automatically inherit the equidistribution property, \cref{lem:equidistribution}. 
In this section we shall prove an equidistribution property which was the 
reason we introduced uniform factors in the first place. To this 
end we will need the following lemma.

\begin{lemma}[Near orthogonality]\label[lemma]{lem:nearorthogonality}
For every decreasing function $\epsilon:\N\rightarrow \R^+$ and 
parameter $d\geq 1$, there is a decreasing function $\gamma:\N\rightarrow \R^+$ 
such that the following holds. Suppose that $\cF=\{P_{i,j}\}_{1\leq i \leq d, 1\leq j\leq M_i}$ 
is a  $\gamma$-uniform factor of degree $d$. Let $x\in \F^n$ be a fixed vector. 
For every integer $k$ and set of not all zero coefficients 
$\Lambda=\{\lambda_{i,j,\omega}\}_{1\leq i\leq d, 1\leq j\leq M_i, \omega\in \{0,1\}^k}$ define
$$
P_\Lambda(y_1...,y_k) \defeq \sum_{i\in [d], j\in [M_i]} \sum_{\omega\in \{0,1\}^k} \lambda_{i,j,\omega}\cdot P_{i,j}(x+ \omega \cdot y),
$$
where $y=(y_1,...,y_k)\in \left(\F^n\right)^k$. Then either we have that $P_\Lambda\equiv 0$ or
$$
\bias(P_\Lambda)<\epsilon(\dim(\cF))
$$
\end{lemma}
\begin{note}
Notice that the statement of this lemma does not require the characteristic $|\F|$ to be large. Such near orthogonality result over systems of linear forms was proved for the large characteristic case, $d<|\F|$, by Hatami and Lovett~\cite{HL11}. Bhattacharyya, et. al.~\cite{BFHHL12} later proved such near orthogonality result for affine systems of linear forms without any assumption on the characteristic of the field. Our proof of \cref{lem:nearorthogonality} uses the derivative technique used in \cite{BFHHL12} except the fact that in our case the linear forms of interest are in $\{0,1\}^k$ makes the proof simpler. 
\end{note}
\begin{proofof}{\cref{lem:nearorthogonality}}
For a vector $\omega\in \{0,1\}^d$, let $|\omega|$ be equal to 
the size of its support. We will say $\omega'\preceq \omega$ if 
$\omega'_i\leq \omega_i$ for every $i$. We may assume that for 
each $i$, $j$, and $\omega$, $\deg(P_{i,j})=i \geq |\omega|$, 
because otherwise we may use the identity
$$
\sum_{\omega'\preceq \omega} (-1)^{|\omega'|}P_{i,j}(x+ \omega'\cdot y) = 0,
$$
to replace $P_{i,j}(x + \omega \cdot y)$ with $P_{i,j}(x+\omega'\cdot y)$'s where 
$|\omega'|<|\omega|$. Notice that $P_{\Lambda'}\equiv 0$ 
after this process is equivalent to $P_\Lambda\equiv 0$. 

Let $\omega^*$ be of maximum support among all $\omega$ for 
which there is an $i$ such that $\lambda_{i,\omega}$ is nonzero. 
Roughly speaking, we will derive $P_\Lambda$ in a way that terms 
of the form $\lambda_{i,j,\omega^*}\cdot P_{i,j}$ will be the only terms 
that can survive. Without loss of generality assume that $|\omega^*|=r$ 
and that $\{1,...,r\}$ is the support of $\omega^*$. We will consider 
$$
DP_\Lambda(y_1,...,y_k,h_1,...,h_r) \defeq D_{h_1}D_{h_2}\cdots D_{h_r} P_\Lambda(y_1,...,y_d),
$$
where $h_i\in (\F^n)^d$ is only supported on $y_i$'s entries. 
It follows from the maximality of $\omega^*$ that
$$
DP_\Lambda (y_1,...,y_k,h_1,...,h_r) = \sum_{i\in [d], j\in [M_i]} \lambda_{i,j,\omega^*} \cdot 
D_{h'_1}\cdots D_{h'_r} P_{i,j}(x+ \omega^* \cdot y),
$$
where each $h'_i$ is the restriction of $h_i$ to $y_i$'s entries, and therefore by a simple 
change of variables we get
\begin{equation}
\label{eq:biasofdp}
\bias(DP_\Lambda) = \norm{e\left( \sum_{i,j} \lambda_{i,j,\omega^*} \cdot P_{i,j}\right)}_{U^r}^{2^r},
\end{equation}
On the other hand we have the following claim which is a result of $r$ 
repeated applications of Cauchy-Schwarz inequality
\begin{claim}\label[claim]{claim:repeatedcauchy}
$$
\bias(P_\Lambda)^{2^r} \leq \bias(DP_\Lambda).
$$
\end{claim}
\begin{proof}
It suffices to show that we have  
$$
\left|\E_{y_1,\ldots,y_k, h\in \F^n} e_\F\big(P_\Lambda(y_1+h, y_2,\ldots,y_k)-P_\Lambda(y_1,y_2,\ldots,y_k)\big)\right| \geq 
\left|\E_{y_1,\ldots,y_k} e_\F\big( P_{\Lambda}(y_1,\ldots,y_k)\big)\right|^2.
$$
This is a simple application of Cauchy Schwarz inequality
\begin{align*}
\left|\E_{y_1,\ldots,y_k, h} e_\F\big(P_\Lambda(y_1+h, y_2,\ldots,y_k)-P_\Lambda(y_1,\ldots,y_k)\big)\right| &= 
\E_{y_2,\ldots,y_k} \left| \E_{z} e_\F\big(P_\Lambda(z,y_2,\ldots,y_k)\big)\right|^2\\ &\geq 
\left| \E_{z,y_2,\ldots,y_k} e_\F(P_\Lambda(z,y_2,\ldots,y_k\in \F^n))\right|^2.
\end{align*}
\end{proof}
Now letting $\gamma(x):= \epsilon(x)$, the proof follows by (\ref{eq:biasofdp}), the definition of
$\gamma$-uniformity and \cref{claim:repeatedcauchy}. 
\end{proofof}

Let $\cF=\{P_{i,j}\}_{1\leq i \leq d, 1\leq j \leq M_i}$ be a polynomial 
factor of degree $d$, and suppose we have a parallelepiped 
$(x+ \omega\cdot y)_{\omega\in \{0,1\}^k}$, where $y=(y_1,...,y_k)\in \left(\F^n\right)^k$. 
The vector $(\cF(x+\omega\cdot y))_{\omega\in \{0,1\}^k}$ can have some dependencies, 
and we cannot expect them to be uniformly distributed similar to \cref{lem:equidistribution}. 
For example if $P$ is a linear function then $\sum_{\omega\in \{0,1\}^2} P(x+\omega\cdot y)\equiv 0$ for every $x$. Therefore the image of $\cF$ over a parallelepiped of dimension $k$ simply 
cannot take every possible value over $\F^{nk}$. 

\paragraph{Consistence.}
Call a vector $b=(b_{i, j ,\omega})_{1\leq i \leq d, 1\leq j \leq M_i, \omega\in \{0,1\}^k}$, where $b\in \F^{\dim(\cF)2^k}$, to be consistent with the factor $\cF$ if the following holds. For every set of coefficients $\Lambda=(\lambda_{i,j, \omega})_{1\leq i \leq d, 1\leq j \leq M_i, \omega\in \{0,1\}^k}$, 
$$
\sum_{i\in [d], j\in [M_i]} \sum_{\omega\in \{0,1\}^k} \lambda_{i,j,\omega}\cdot P_{i,j}(x+\omega\cdot y) \equiv 0,
$$
implies
$$
\sum_{i\in [d], j\in [M_i]} \sum_{\omega\in \{0,1\}^k} \lambda_{i,j,\omega} \cdot b_{i,j,\omega}=0.
$$

Let $\Sigma_{\Box}$ be the set of all vectors that are consistent 
with $\cF$. It is easy to verify that $\Sigma_{\Box}$ forms a vector space. The following lemma of \cite{GT09} computes the dimension of $\Sigma_\Box$.

\begin{lemma}[\cite{GT09}]\label[lemma]{lem:sigmabox}
Suppose that $k>d$. Then we have 
$$
\dim(\Sigma_\Box) = \sum_{i\in [d]} \sum_{0\leq j\leq i} {k \choose j}.
$$
\end{lemma}

Now the proof of the equidistribution over parallelepipeds follows in a straightforward manner.

\begin{lemma}[Equidistribution of parallelepipeds]\label[lemma]{lem:parallelequidistribution}
Suppose that $\epsilon:\N\rightarrow \R^+$ is a decreasing 
function, and $k>d>0$ are integers. Suppose that $\cF=\{P_{i,j}\}_{1\leq i\leq d, 1\leq j\leq M_i}$ 
is a $\gamma_{\ref{lem:nearorthogonality}}$-uniform factor of degree $d$. Let $x\in \F^n$ 
be a fixed vector, $b=(b_{i,j,\omega})_{1\leq i\leq d, 1\leq j\leq M_i, \omega\in \{0,1\}^k}\in \F^{\dim(\cF)2^k}$ 
is consistent with $\cF$, and moreover $P_{i,j}(x)= b_{i,j,\emptyset}$ for all $i$ and $j$. Then we have

$$
\Pr_{x,y_1,...,y_k}\left[ P_{i,j}(x+ \omega\cdot y)= b_{i,j,\omega}\; \forall i\in [d], \forall j\in[M_i], \forall \omega\in\{0,1\}^k  \right] = (1\pm \epsilon) \frac{1}{|\Sigma_{\Box}|}.
$$
\end{lemma}
\begin{proof}
Similar to the proof of \cref{lem:equidistribution} we have
\begin{align*}
\Pr_{y_1,...,y_k}\left[ P_{i,j}(x+ \omega\cdot y)= b_{i,j,\omega}\; \forall i,j,\omega \right] &= 
\E\left[ \prod_{i,j,\omega} \frac{1}{|\F|} \left(\sum_{\lambda_{i,j,\omega}=0}^{|\F|-1} e_\F\bigg(\lambda_{i,j,\omega} (P_{i,j}(x+ \omega\cdot y) - b_{i,j,\omega})\bigg)\right)\right]\\
&= \frac{1}{|\F|^{\dim(\cF)2^k}}\sum_{\Lambda} e_\F\bigg(\sum_{i,j,\omega} \lambda_{i,j,\omega} b_{i,j,\omega}\bigg)\E\left[ e_\F\bigg(\sum_{i,j,\omega} \lambda_{i,j,\omega}P_i(x+ \omega \cdot y)\bigg) \right]\\
&= \frac{1}{|\Sigma_\Box|}\pm \epsilon(\dim(\cF),
\end{align*}
where $i$ is always taken from $[d]$, $j$ from $[M_i]$, $\omega$ 
from $\{0,1\}^k$, and $\Lambda=(\lambda_{i,j,\omega})$ is a collection 
of coefficients from $\F$. The last inequality follows from \cref{lem:nearorthogonality}.
\end{proof}

\subsection{Proof of Proposition \ref{prop:approxtocompute_new}}
\label{sec:approxtocompute_new}
Here we will present how to adapt the proof of Proposition~5.1 from \cite{GT09} to our setting.

\begin{proofof}{of \cref{prop:approxtocompute_new}} 
Let $X$ be the set of points in $\F^n$ for which $P(x)=\tP(x)$. 
By assumption we know that $|X|\geq (1-\sigma_d)|\F|^n$. We will 
use $(d+1)$-dimensional parallelepipeds to prove that $P(x)=\tP(x)$ 
within the whole of most of the atoms of $\cF$, and then to prove 
that $P$ should be constant on each remaining atom as well. 

By \cref{lem:equidistribution} we have that for $1-O(\sqrt{\sigma_d})$ 
of the atoms $A$ of $\cF$ we have $\Pr_{x\in A}\left(P(x)=\tP(x) \right)\geq 1-O(\sqrt{\sigma_d})$; 
we will say that these atoms are \emph{almost good}. 
We will prove that almost good atoms are in fact, totally good. 

Suppose that $A=\cF^{-1}(t)$ is an almost good atom, where
$t=(t_{i,j})_{1\leq i \leq d, 1\leq j\leq M_i}$ is a vector in 
$\F^{\dim(\cF)}$. Let $A'\defeq A\cap X$ be the set of points in 
$A$ for which $P(x)=\tP(x)$. The following lemma proved in \cite{GT09} 
can be adapted to our setting by \cref{lem:equidistribution} 
and \cref{lem:parallelequidistribution}.

\begin{lemma}[Lemma 5.2 of \cite{GT09}]\label[lemma]{lem:goodatoms}
Suppose that $\sigma_d$ is a sufficiently small constant. Fix an $x\in A$. 
Then there is $h$ so that all of the vertices $x+\omega\cdot h$, 
$\omega\neq 0^{d+1}$, lie in $A'$. 
\end{lemma}
 
Let $x$ be an arbitrary point in $A$, and let $h$ be as in \cref{lem:goodatoms}. 
Since $P$ is a degree $d$ polynomial, we have the following derivative identity
$$
P(x) = -\sum_{\omega\in \{0,1\}^{d+1}, \omega\neq 0} (-1)^{|\omega|} P(x+\omega\cdot h)= -\sum_{\omega\in \{0,1\}^{d+1}, \omega\neq 0} (-1)^{|\omega|} \tP(x+\omega\cdot h).
$$
Since $\tP$ is constant on $A$, this means that $P(x)= \tP(x)$. 
Since $x$ was an arbitrary element of $A$, this implies that $P$ is constant on $A$. 

We have proved that $1-\rO(\sqrt{\sigma_d})$ of the atoms of $\cF$ 
are totally good. Let $A$ be an arbitrary atom of $\cF$. The following claim 
will allow us to show that $A$ is totally good.

\begin{claim}
\label[claim]{claim:badatomsaregood}
Assume that $1-\rO(\sqrt{\sigma_d})$ of the atoms are totally good, where $\sigma_d\leq c2^{-d}$. 
Then for every atom $A$, there is a set of good atoms $A_\omega$ for $\omega\neq 0$ in 
$\{0,1\}^{d+1}$ such that for every $x\in A$ there exists $h=(h_1,...,h_{d+1})$ 
such that $x+\omega \cdot h \in A_\omega$. 
\end{claim}
\begin{proof}
Notice that by \cref{lem:parallelequidistribution} it suffices 
to prove this fact for one arbitrary point $x$ in $A$. We will pick 
$h_1,...,h_k$ uniformly at random from $\F^n$. By \cref{lem:equidistribution}, 
for each $\omega\neq 0$, the probability that $x+\omega\cdot h$ lies in a good 
atom is $1-O(\sqrt{\sigma_d})$. Choosing $\sigma_d\leq c2^{-d}$ for sufficiently 
small constant $c$, it follows by union bound that with positive probability 
there is a choice of $h=(h_1,...,h_{d+1})$ such that $x+\omega\cdot h$ is a good atom. 
\end{proof}

Existence of good atoms $A_\omega$ as above implies that $P$ has to 
be constant on $A$ as well, again using the fact that $P$ is a degree $d$ polynomial and therefore 
$$
P(x) = -\sum_{\omega\in \{0,1\}^{d+1}, \omega \neq 0} (-1)^{|\omega|}P(x+\omega\cdot h).
$$

This completes the proof of \cref{prop:approxtocompute_new}.
\end{proofof}

\section{Applications in High Characteristics}
In this section we show some applications of the uniformity notion, and our algorithm 
for uniform refinement of a factor. Throughout the section we shall assume that the 
field $\F$ has large characteristic, namely $|\F|$ is greater than the degree of the 
polynomials and the degree of the factors that we study. 

\subsection{Computing Polynomials with Large Gowers Norm}
Following is an immediate corollary of \cref{lem:algorithmicBV}, 
\cref{lem:query}, and \cref{prop:approxtocompute_new}. 

\begin{lemma}
\label[lemma]{lem:biasimpliesrankhighchar}
Suppose that an integer $d$ satisfies $0\leq d < |\F|$. Let $\delta, \sigma, \beta \in (0,1]$. 
There is a randomized algorithm that given query access to a polynomial 
$P:\F^n\rightarrow \F$ of degree $d$ such that $\bias(P)\geq \delta >0$, runs 
in $O_{\delta, \beta, \sigma}(n^d)$ and with probability $1-\beta$, returns 
a polynomial factor $\cF=\{ P_{i,j} \}_{1\leq i\leq d-1, 1\leq j\leq M_i}$ of 
degree $d-1$ and $\dim(\cF)= O_{\delta, \beta, \sigma}(1)$ such that $P$ is 
measurable in $\cF$, namely there is a function $\Gamma$ such that
$$
P(x)= \Gamma\left( \left(P_{i,j}(x)\right)_{1\leq i \leq d-1, 1\leq j \leq M_i}\right).
$$

Moreover for every query to the value of $\Gamma(\cdot)$ for a vector from $\F^{\dim(\cF)}$ 
the algorithm returns the correct answer with probability $1-\beta$.
\end{lemma}
\begin{proof}
Let $\sigma_d$ and $\gamma_1$ be as in \cref{prop:approxtocompute_new}, 
and let $\gamma_2$ be as in \cref{lem:query}. Set $\gamma_d:=\min\{\gamma_1,\gamma_2\}$. 
By \cref{lem:algorithmicBV} since $\bias(P)\geq \delta$, in time $O_{\sigma_d, \beta}(n^d)$, 
with probability $1-\frac{\beta}{2}$ we can find a polynomial factor $\cF$ of degree $d-1$  such that 
\begin{itemize}
\item $P$ is $\frac{\sigma_d}{2}$-close to a function of $\cF$
\item $\dim(\cF)\leq \frac{4|\F|^5}{\delta^2 \sigma_d \beta}$
\end{itemize}
By \cref{lem:uniformrefinement}, with probability $1-\frac{\beta}{2}$, 
we can find a $\gamma_d$-uniform factor $\cF'$ such that $\cF'$ is 
$\frac{\sigma_d}{2}$-close to being a refinement of $\cF$. 

Thus with probability greater than $1-\beta$, $P$ is $\sigma_d$-close 
to a function of the strongly $\gamma_d$-unbiased factor $\cF'$, and 
by \cref{lem:approxtocomputelowchar_new} 
$P$ is measurable in $\cF'$. The query access to $\Gamma$ can be granted 
by \cref{lem:query}.
\end{proof}

The following is an algorithmic version of Proposition~6.1 of \cite{GT09}, which states that 
given a polynomial $P$ such that $e_\F(P)$ has large Gowers norm, then we can compute $P$ by 
a few lower degree polynomials. 

\begin{prop}[Computing Polynomials with High Gowers Norm]
\label[prop]{prop:gowersimplieslowrank}
Suppose that $|\F|>d\geq 2$ and that $\delta, \beta\in (0,1]$. There is a 
randomized algorithm that given query access to a polynomial $P:\F^n \rightarrow \F$ 
of degree $d$ with $\norm{e_\F(P(x))}_{U^d}\geq \delta$, runs in $\rO_{\delta, \beta}(n^d)$ 
and with probability $1-\beta$, returns a polynomial factor $\cF$ of degree $d-1$ such that
\begin{itemize}
\item There is a function $\Gamma:\F^{\dim(\cF)}\rightarrow \F$ such that $P=\Gamma(\cF)$.
\item $\dim(\cF)=O_{d,\delta,\beta}(1)$.
\end{itemize}

Moreover, for every query to the value of $\Gamma(\cdot)$ for a vector from $\F^{\dim(\cF)}$, 
the algorithm returns the correct answer with probability $1-\beta$.
\end{prop}
\begin{proof}
Write $\partial^dP(h_1,...,h_{d}) := D_{h_1}\cdots D_{h_{d}}P(x)$. Since $P$ 
has degree $d$, $\partial^d P$ does not depend on $x$. From the definition 
of the $U^d$ norm, we have
$$
\bias(\partial^d P)= \norm{e(P)}_{U^d}^{2^d} \geq \delta^{2^d}.
$$

Applying \cref{lem:biasimpliesrankhighchar} to $\partial^{d}P$, 
with probability $1-\frac{\beta}{2}$, we can find a factor $\widetilde{\cF}$ 
of degree $d-1$, such that $\dim(\widetilde{\cF})=O_{\delta, \beta, d}(1)$ 
and $\partial^dP$ is measurable in $\widetilde{\cF}$.

It is easy to check that since $|\F|>d$, we have the following Taylor expansion
$$
P(x) = \frac{1}{d!} \partial^{d} P(x,...,x) + Q(x),
$$
where $Q$ is a polynomial of degree $\leq d-1$. We can find explicit 
description of $P$, and therefore $Q$ in $O(n^d)$. Let $\cF':= \widetilde{\cF} \cup \{Q\}$, 
and let $\gamma$ be as in \cref{lem:query} for error parameter $\beta$. 
By \cref{lem:uniformrefinement}, with probability $1-\frac{\beta}{2}$, 
we can refine $\cF'$ to a $\gamma$-uniform factor $\cF$ of same degree $d-1$ 
with $\dim(\cF)=O_{\dim(\cF'), \beta, \gamma}(1)$. Now query access can be 
granted by \cref{lem:query}.
\end{proof}

\subsection{Decoding Reed-Muller Codes with Structured Noise}
\label{sec:reedmuller}
In this section we discuss the task of decoding Reed-Muller codes when the noise is ``structured''. 
Recall that the codewords in a Reed-Muller code of order $k$, correspond to evaluations of all 
degree $k$ polynomials in $n$ variables over $\F$. We will consider the task of decoding codewords 
when the noise is structured in the sense that the applied noise to the codeword itself is a polynomial 
of higher degree $d$. In other words, we are given a polynomial $P$ of degree $d>k$ with the 
promise that it is ``close'' to an unknown degree $k$ polynomial, and the task is to find a degree $k$ polynomial $Q$ 
that is somewhat close to $P$. The assumption of $d>k$ is made, because the case $d\leq k$ is trivial.

\begin{theorem}[Reed-Muller Decoding]\label{thm:rmdecoding}
Suppose that $d> k>0$ are integers, and let $\epsilon\in (0,1]$ and 
$\beta\in (0,1]$ be parameters. There is $\delta(\epsilon,d,k)>0$ and a 
randomized algorithm that, given query access to a polynomial $P:\F^n\rightarrow \F$ 
of degree $d$, with the promise that there exists a polynomial $Q$ of 
degree $k$ such that $\Pr(P(x)=Q(x))\geq \frac{1}{|\F|}+\epsilon$, runs in $O(n^d)$, 
and with probability $1-\beta$ returns a polynomial $\tQ$ of degree $k$ such that
$$
\Pr(P(x)\neq \tQ(x))\geq \frac{1}{|\F|}+ \delta.
$$
\end{theorem}

Before presenting the proof of \cref{thm:rmdecoding} we will look at the relations between Hamming distance of $P$ and $Q$, correlation between $e_\F(P)$ and $e_\F(Q)$, and Gowers norm of $e_\F(P)$. The following claim relates the Hamming distance of two polynomials 
$P$ and $Q$ to the correlation between $e_\F(P)$ and $e_\F(Q)$. 
\begin{claim}
\label[claim]{claim:distancetocorrelation}
Suppose that $\epsilon\in (0,1]$, and let $P,Q:\F^n \rightarrow \F$ be two functions 
such that $\Pr_x(P(x)=Q(x)) \geq 1/|\F|+\epsilon$. Then there exists a nonzero $t\in \F$ for which
$$
\big|\langle e_\F(t\cdot P), e_\F(t\cdot Q)\rangle\big| = \left| \E_x e_\F\big(t\cdot (P(x)-Q(x))\big) \right|\geq \epsilon.
$$
\end{claim}
\begin{proof}
We have
\begin{align*}
\frac{1}{|\F|}+ \epsilon\leq \Pr_{x\in \F^n}\big( P(x)=Q(x) \big) &= \E_{x\in \F^n}\big[ 1_{P(x)=Q(x)}\big] \\ &=
\E_{x\in \F^n}\left[\frac{1}{|\F|} \sum_{t\in \F} e_\F\big(t\cdot(P(x)-Q(x))\big )\right]\\ &=
\frac{1}{|\F|} \sum_{c\in \F} \langle e_\F(t\cdot P), e_\F(t\cdot Q)\rangle\\ &=
\frac{1}{|\F|}\left(1 + \sum_{t\in \F\backslash \{0\}} \ip{e_\F(t\cdot P), e_\F(t\cdot Q)} \right),
\end{align*}
and thus there is $t\in \F\backslash \{0\}$ for which $ \big|\ip{e_\F(t\cdot P), e_\F(t\cdot Q)} \big| \geq \epsilon$. 
\end{proof}

The following lemma which was first proved in~\cite{Gow01} is a key property of Gowers norms. 
\begin{lemma}[Gowers Cauchy-Schwarz]
\label[lemma]{lem:GowersCauchySchwarz}
Let $G$ be a finite Abelian group, and consider a family of functions $f_S:G \rightarrow \mathbb{C}$, where $S \subseteq [k]$. Then
\begin{equation}
\left| \E_{X,Y_1,\ldots,Y_d\in \F^n} \left[\prod_{S \subseteq [d]} \mathcal{C}^{d-|S|} f_S(X + \sum_{i \in S} Y_i) \right]\right|\le \prod_{S \subseteq [d]} \norm{f_S}_{U^d},
\end{equation}
\end{lemma}

Assume that $P:\F^n \rightarrow \F$ is a polynomial of degree $d$ and $Q:\F^n\rightarrow \F$ is a polynomial of degree $k<d$. Using the above lemma with $f_\emptyset:=e_\F(P-Q)$ and $f_S:=1$ for $S\neq \emptyset$ implies that 
\begin{equation}\label{eq:gowerscorr}
|\ip{e_\F(P), e_\F(Q)}| = |\E_{x\in \F^n} e_\F(P(x)-Q(x))| \leq \norm{e_\F(P(x)-Q(x))}_{U^d} = \norm{e_\F(P(x))}_{U^d},
\end{equation}
where the last equality follows from the fact that $\deg(Q(x))=k<d$. 

\begin{proofof}{of \cref{thm:rmdecoding}}
Since there exists a polynomial $Q$ of degree $k$ with $\Pr(P(x)\neq Q(x))\leq \epsilon$, 
by \cref{claim:distancetocorrelation} there exists a nonzero $t\in \F$ such that
$$
\big| \langle e_\F(t\cdot P), e_\F(t\cdot Q) \rangle \big| \geq \epsilon.
$$
Notice that
$$
\norm{e_\F(t\cdot P)}_{U^d}\geq \langle e_\F(t\cdot P), e_\F(t\cdot Q) \rangle \geq \epsilon,
$$
where the first inequality follows from (\ref{eq:gowerscorr}) since $\deg(t\cdot Q)\leq k$. Notice that we may find such a constant $t\in \F\backslash \{0\}$ with high probability by estimating the $U^d$ norm of $tP$ for all $|\F|-1$ possible choices of $t$. Set $\tP:=t\cdot P$. 
By \cref{prop:gowersimplieslowrank}, 
with probability $1-\frac{\beta}{6}$ we can find a polynomial factor $\tcF$ of 
degree $d-1$ such that $\tP$ is measurable in $\tcF$. Let 
$\gamma:\N\rightarrow \R^+$ be a decreasing function to be specified later. 
By \cref{lem:uniformrefinement}, with probability $1-\frac{\beta}{6}$, 
we can refine $\tcF$ to a $\gamma$-uniform factor $\cF=\{P_1,...,P_m\}$ of 
same degree $d-1$, with $\dim(\cF)= \rO_{\gamma, \beta}(1)$. Since $\tP$ is 
measurable in $\cF$, there exists $\Gamma:\F^{\dim(\cF)}\rightarrow \F$ 
such that $\tP= \Gamma(\cF)$. Using the Fourier decomposition of $e_\F(\Gamma)$ we can write

\begin{equation}
\label{eq:fourierdecompose}
f(x)\defeq e_\F\big(\tP(x)\big)= \sum_{i=1}^L c_i \; e_\F\big( \langle \alpha^{(i)}, \cF\rangle (x)\big),
\end{equation}

where $L=|\F|^{\dim(\cF)}=O_{\gamma, \beta}(1)$, $\alpha^{(i)}\in \F^{\dim(\cF)}$, and
$$
\langle \alpha^{(i)}, \cF\rangle(x) \defeq \sum_{j=1}^m \alpha^{(i)}_j\cdot P_j(x).
$$ 
Notice that the terms in (\ref{eq:fourierdecompose}), unlike Fourier characters, 
are not orthogonal. But since the factor is $\gamma$-uniform, \cref{lem:nearorthogonality} 
ensures approximate orthogonality. Let $Q_i:=\langle \alpha^{(i)}, \cF \rangle$. 
Choose $\gamma(u)\leq \frac{\sigma}{|\F|^{2u}}$, so that 
$\gamma(\dim(\cF))\leq \frac{\sigma}{L^2}$, where $\sigma:= \frac{\eps^{2^{k+1}}}{4}$.
It follows from the near orthogonality of the terms in (\ref{eq:fourierdecompose}) by
\cref{lem:nearorthogonality} that
\begin{equation}
\label{eq:fouriercoeff}
|c_i- \langle f, e_\F(Q_i)\rangle| \leq \frac{\sigma}{L},
\end{equation}
and 
\begin{equation}
\label{eq:l2}
\bigg|\norm{f}_2^2 - \sum_{i=1}^L c_i^2\bigg| \leq \sigma.
\end{equation}

\begin{claim}
\label[claim]{claim:largecoeff}
There exists $\delta'(\epsilon, |\F|) \in (0,1]$ such that the following 
holds. Assume that $f$ and $\cF$ are as above. Then there is $i\in [L]$, 
for which $\deg(Q_i)\leq k$ and $\big|\langle f, e_\F(Q_i)\rangle\big| \geq \delta'$.
\end{claim}
\begin{proof}
We will induct on the degree of $\cF$. Assume for the base case 
that $\cF$ is of degree $k$, i.e. $d=k+1$, thus applying the following Cauchy-Schwarz inequality
\begin{equation}
\label{eq:gowersl2}
\epsilon^{2^{k+1}}\leq \norm{f}_{U^{k+1}}^{2^{k+1}} \leq \norm{f}_2^2 \norm{f}_{\infty}^{2^{k+1}-2}, 
\end{equation}
and (\ref{eq:l2}) imply that there exists $i\in [L]$ such that 
$c_i^2 \geq \frac{\epsilon^{2^{k+1}} - \sigma }{L}= \frac{3\epsilon^{2^{k+1}}}{4L}$, 
which combined with (\ref{eq:fouriercoeff}) implies that 
$|\langle f, e_\F(Q_i)\rangle| \geq \frac{\epsilon^{2^{k+1}}}{2L}$.

Now for the induction step, assume that $d>k+1$. We will decompose 
(\ref{eq:fourierdecompose}) into two parts, first part consisting 
of the terms of degree $\leq k$ and the second part consisting of 
the terms of degree strictly higher than $k$. Namely, letting 
$S:=\{i\in [L]: \deg(Q_i)\leq k\}$ we write $f=g+h$ where 
$g\defeq \sum_{i\in S}c_i e_\F(Q_i)$ and $h\defeq \sum_{i\in [L]\backslash S}c_i e_\F(Q_i)$. 
Notice that by the triangle inequality of Gowers norm, our choice of 
$\gamma$, and the fact that $\cF$ is $\gamma$-uniform
$$
\norm{h}_{U^{k+1}}\leq \sum_{i\in [L]\backslash S} |c_i|\cdot \norm{e_\F(Q_i)}_{U^{k+1}}\leq  L\cdot \frac{\epsilon^{2^{k+1}}}{4L^2} = \frac{\epsilon^{2^{k+1}}}{4L},
$$ 
and thus
$$
\norm{g}_{U^{k+1}} \geq \frac{\epsilon}{2}.
$$
Now the claim follows by the base case. 
\end{proof}

Let $\delta'(\epsilon, |\F|)$ be as in the above claim. 
We will use the following theorem of Goldreich and Levin~\cite{GL89} 
which  gives an algorithm to find all the large Fourier coefficients of $e_\F(\Gamma)$. 

\begin{theorem}[Goldreich-Levin theorem~\cite{GL89}]
Let $\zeta, \rho \in (0,1]$. There is a randomized algorithm, 
which given oracle access to a function $\Gamma:\F^m \rightarrow \F$, 
runs in time $\rO\big(m^2\log m \cdot \mathrm{poly}(\frac{1}{\zeta}, \log(\frac{1}{\rho}))\big)$ 
and outputs a decomposition
$$
\Gamma= \sum_{i=1}^\ell b_i \cdot e_\F(\langle \eta_i, x\rangle) + \Gamma',
$$
with the following guarantee:
\begin{itemize}
\item $\ell=\rO(\frac{1}{\zeta^2})$.
\item $\Pr \big[ \exists i: \; |b_i- \hat{\Gamma}(\eta_i)|> \zeta/2 \big]\leq \rho.$
\item $\Pr \bigg[ \forall \alpha \text{ such that } |\hat{f}(\alpha)|\geq \zeta, \exists i\; \eta_i=\alpha\bigg] \geq 1-\rho.$
\end{itemize}
\end{theorem}

We will use the above theorem with parameters $\zeta:= \frac{\delta'}{2}$ and 
$\rho:=\frac{\beta}{6}$, and use the randomized algorithm in \cref{lem:query} 
to provide answer to all its queries to $\Gamma$. Choose $\gamma$ suitably small, 
so that with probability at least $1-\frac{\beta}{6}$ our answer to all queries 
to $\Gamma$ are correct. By \cref{claim:largecoeff} there is $i\in [L]$ such that 
$\hat{\Gamma}(\alpha)=c_i\geq \frac{3\delta'}{4}$. With probability $1-\frac{\beta}{6}$ 
there is $j$ such that $\eta_j=\alpha_i$ and with probability at least $
1-\frac{\beta}{6}$, $|b_j-c_i|\leq \frac{\zeta}{2} \leq \frac{\delta'}{4}$, 
and therefore $c_i\geq \frac{\delta'}{2}$.

By a union bound, adding up the probabilities of the errors, 
with probability at least $1-\frac{5\beta}{6}$, we find $Q_i$  such that 
$$
\big|\langle f, e_\F(Q_i)\rangle\big| \geq \frac{\delta'}{4}.
$$

The following claim shows that, there is a constant shift of $e_\F(Q_i)$ that approximates $f$.

\begin{claim}
Recall that $f(x)=e_\F(\tP(x))=e_\F(t\cdot P(x))$, and we have 
$$|\langle e_\F(\tP), e_\F(Q_i) \rangle| = \big|\E_{x\in \F^n} e_\F(\tP(x)-Q_i(x))\big| \geq \frac{\delta'}{4}.$$

There is a randomized algorithm that with probability $1-\frac{\beta}{6}$ returns an $h\in \F$ for which
$$
\Pr(\tP(x)= Q_i(x)+h) \geq \frac{1}{|\F|}+\frac{\delta'}{8|\F|^2}
$$
\end{claim} 
\begin{proof}
Since $\tP(x)-Q_i(x)$ takes values in $\F$, there must be a choice of $r\in \F$ such that $\Pr(\tP(x)-Q_i(x)=r)\geq \frac{1}{|\F|}+ \frac{\delta'}{4|\F|^2}$. 
Similar to the proof of \cref{lem:algorithmicBV} defining $\mu_0(r)= \Pr(\tP(x)-Q_i(x)=r)$, 
we can find the estimate measure $\mu_{\obs}$ by $K=\rO_{|\F|, \delta', \beta}(1)$ 
random queries to $\tP(x)-Q_i(x)$ such that
$$
\Pr \bigg(\exists r\in \F:\; \big|\mu_0(r) - \mu_{\obs}(r)\big|\geq \frac{\delta'}{8|\F|^2}\bigg)\leq \frac{\beta}{6}.
$$

Let $h\in \F$ be such that $\mu_{\obs}(h)$ is maximized. Letting $Q':= Q_i+h$ we have 
\begin{itemize}
\item $\Pr \big( \tP(x) = Q'(x) \big)\geq \frac{1}{|\F|}+ \frac{\delta'}{8|\F|^2}$
\item $\deg(Q')= k$. 
\end{itemize}
\end{proof}

Since $\tP=t\cdot P$, the same also holds between $P$ and $t^{-1}Q'$. 
The probability that all steps of the algorithm work correctly is 
bounded from below by $1-\frac{5\beta}{6}-\frac{\beta}{6} =1-\beta$.
\end{proofof}

\section{Low Characteristics}\label{sec:lowchar}
As mentioned in \cref{sec:uniformrefinement}, notions of 
regularity (\cref{dfn:regularfactor}) and uniformity 
(\cref{dfn:uniformfactor}) are not sufficient to prove 
\cref{prop:approxtocompute} and \cref{prop:approxtocompute_new}, 
respectively, without the assumption on $|\F|$ being larger than $d$. 
In this section we will use strongly unbiased factors (\cref{dfn:stronglyunbiased}) 
to handle the case of low characteristics. 

\restate{\cref{dfn:stronglyunbiased}}{
Suppose that $\gamma:\N\rightarrow \R^+$ is a decreasing function. 
Let $\cF=\{P_1,...,P_m\}$ be a polynomial factor of degree $d$ over $\F^n$ and let 
$\Delta:\cF\rightarrow \N$ assign a natural number to each polynomial 
in the factor. We say that $\cF$ is strongly $\gamma$-unbiased with degree bound $\Delta$ if 
\begin{enumerate}
\item For every $i\in [m]$, $1\leq \Delta(P_i)\leq \deg(P_i)$. 
\item For every $i\in [m]$ and $r>\Delta(P_i)$, there exists a function $\Gamma_{i,r}$ such that
$$
P_i(x+y_{[r]})= \Gamma_{i,r}\left(P_j(x+y_J): j\in [m], J\subseteq [r], |J|\leq \Delta(P_j)\right).
$$
\item For every $r\geq 0$ and collection of coefficients 
$\cC=\{c_{i,I}\in \F\}_{i\in [m], I\subseteq [r], |I|\leq \Delta(P_i)}$ 
not all of which are zero, we have
$$
\bias\left( \sum_{i\in [m], I\subseteq [r], |I|\leq \Delta(P_i)} c_{i,I}P_i(x+y_I)\right) \leq \gamma(\dim(\cF)),
$$
for random variables $x,y_1,\ldots,y_r\in \F^n$.
\end{enumerate}
We will say $\cF$ is $\Delta$-bounded if it only satisfies (1) and (2). Moreover for every $r>0$ let
$$
B(r)\defeq \sum_{1\leq i\leq m} \sum_{0\leq j\leq \Delta(P_i)} {r \choose j},
$$
be the number of pairs $(i,I)$ with $i\in [m]$, $I\subseteq [r],$ and $|I|\leq \Delta(P_i)$.
}

The following claim states that although in part (3) there is no upper bound on $r$,  
choosing a suitably faster decreasing function $\gamma'= \gamma^{2^d}$, one can assume that $r\leq d$. 

\begin{claim}
Let $\cF=\{P_1,\ldots, P_m\}$ be a polynomial factor of degree $d$, and 
$\Delta:\cF\rightarrow \N$ be a function such that $\cF$ is $\Delta$-bounded. 
Suppose that $\gamma:\N\rightarrow \R^+$ is a decreasing function such that
for every $r\in [d]$, $x\in \F^n$ and not all zero collection of coefficients 
$\{c_{i,I}\}_{i\in [m], I\subseteq [r], |I|\leq \Delta(P_i)}$
\begin{equation}\label{eq:stronglyunbiased}
\bias\left(\sum_{i\in [m], I\subseteq [r], |I|\leq \Delta(P_i)} c_{i,I}\cdot P_i(x+y_I) \right) \leq \gamma(\dim(\cF))^{2^d},
\end{equation}
where $y_1,\ldots, y_r$ are variables from $\F^n$. Then $\cF$ is strongly $\gamma$-unbiased. 
\end{claim}
\begin{proof}
Let $r>d$ be an integer, $x\in \F^n$ be fixed and $\{c_{i,I}\}_{i\in [m], I\subseteq [r], |I|\leq \Delta(P_i)}$ be a set of coefficients, not all of which are zero. We will show that 
$$
\bias\left(\sum_{i\in [m], I\subseteq [r], |I|\leq \Delta(P_i)} c_{i,I}\cdot P_i(x+y_I) \right) \leq \gamma(\dim(\cF)).
$$
Define $Q_\cC(y_1,\ldots,y_r)\defeq \sum_{i\in [m], I\subseteq [r], |I|\leq \Delta(P_i)} c_{i,I}\cdot P_i(x+y_I)$, 
and let $I\subseteq [r]$ be maximal such that $c_{i,I}$ is nonzero for some $i\in [m]$. 
Without loss of generality assume that $I=\{1,2,\ldots, |I|\}$. We will derive 
$Q_\cC$ in directions $y_1,\ldots,y_{|I|}$, namely
$$
DQ_{\cC,y_1,\ldots,y_r}(h_1,\ldots,h_{|I|}) \defeq D_{h_1}\cdots D_{h_{|I|}} Q_\cC(y_1,\ldots,y_r),
$$
where $h_i$ is only supported on $y_i$ coordinates. Notice that, since $I$ was maximal, 
$DQ_\cC$ does not depend on any $y_i$ with $i\in [r]\backslash I$, namely, for every such $J\subseteq [r]$ with $J\not\subseteq I$, and every $i\in [m]$, $D_{h_1}\cdots D_{h_{|I|}} P_i(x+y_J) \equiv 0$. Moreover
\begin{align*}
D_{h_1}\cdots D_{h_{|I|}} Q_\cC(y_1,\ldots,y_r)&= 
\sum_{i\in [m], J\subseteq [r], |J|\leq \Delta(P_i)} c_{i,J} \cdot D_{h_1}\cdots D_{h_{|I|}} P_i(x+y_J)\\ &=
\sum_{i\in [m]} c_{i,I} \cdot  D_{h_1}\cdots D_{h_{|I|}} P_i(x+y_I) \\ &=
\sum_{i\in [m], J\subseteq I} c_{i,I} \cdot (-1)^{|J|}\cdot  P_i((x+y_I)+ h_{J}).
\end{align*}
It follows by \cref{claim:repeatedcauchy} for the random variable $x':=x+y_I$ and taking the bias over random variables $x', h_1,\ldots,h_{|I|}$ that
$$
\gamma(\dim(\cF))^{2^{d}}\geq \bias\left(DQ_{\cC,y_1,\ldots,y_r}(x',h_1,\ldots,h_{|I|})\right) \geq \bias\left(Q_\cC(y_1,\ldots,y_r)\right)^{2^{|I|}}.
$$
Now the claim follows by the fact that $|I|\leq d$. 
\end{proof}

The above claim suggests a way of refining a factor to a strongly 
uniform one: We will estimate the bias of every possible linear 
combination of the polynomials over all points $x+y_I$ where 
$I\subseteq [r]$ and $r\leq d$, and once we have detected a biased combination, refine 
using \cref{lem:algorithmicBV}. This is made possible by the 
fact that $r$ and $\dim(\cF)$ are both bounded.

\begin{lemma}[Algorithmic strongly $\gamma$-unbiased refinement]
\label[lemma]{lem:stronglyunbiasedrefinement}
Suppose that $\sigma, \beta \in (0,1]$ are parameters, and 
$\gamma:\N \rightarrow \R^+$ is a decreasing function. There is 
a randomized algorithm that given a factor $\cF=\{P_1,...,P_m\}$ 
of degree $d$, runs in $\rO_{\gamma}(n^d)$, with probability $1-\beta$, 
returns a strongly $\gamma$-unbiased factor $\cF'$ with degree bound 
$\Delta$ such that $\cF'$ is $\sigma$-close to being a refinement of $\cF$. 
\end{lemma}
\begin{proof}
We will design a refinement process which has to stop in a finite 
number of steps which only depends on $\dim(\cF)$, $\sigma$, $\beta$, 
and $\gamma$. We will start with the initial factor being $\cF$ and 
we will set $\Delta(P_i):= \deg(P_i)$. Notice that this automatically 
satisfies the first two properties of \cref{dfn:stronglyunbiased}, 
namely $\Delta$-boundedness. This is because for every $r>\deg(P_i)$ we can use 
the derivative property
$$
P_i(x+y_{[r]}) = \sum_{I\subsetneq [r]} (-1)^{r-|I|}\cdot P_i(x+y_I). 
$$
to write $P_i(x+y_{[r]})$ as a function of $\{P_i(x+y_J)\}_{J\subseteq [r], |J|\leq d}$.

As explained in proof of \cref{lem:uniformrefinement}, at each step, 
we may assume without loss of generality that all polynomials in the factor 
are homogeneous. Moreover, we will assume that for every $r\leq d$, 
the set of polynomials $\{P_{i}(x+y_I)\}_{i\in [m], I\subseteq [r], |I|\leq \Delta(P_i)}$ 
is linearly independent, where $x,y_1,\ldots, y_r\in \F^n$ are variables. This is because 
at each step we have access to explicit description of the polynomials in the factor, 
therefore we can compute each of the $|\F|^{B(r)}$ possible linear combinations, where 
$B(r)=\sum_{i\in [m]}\sum_{j\in [\Delta(P_i)]} {r \choose j}$, 
and check whether it is equal to zero or not. Suppose that for a set of coefficients 
$\cC=\{c_{i,I}\in \F\}_{i\in [m], I\subseteq [r], |I|\leq \Delta(P_i)}$ 
we have $Q_\cC\defeq \sum_{i\in [m], I\subseteq [r], |I|\leq \Delta(P_i)}c_{i,I}P_i(x+y_I) \equiv 0$. 
Let $P_i$ be a polynomial that appears in $Q_\cC$ and let $I$ be maximal such that $c_{i,I}\neq 0$. 
We can let $\Delta(P_i):= |I|-1$, and remove $P_i$ from the factor if $\Delta(P_i)$ becomes zero.

We will stop refining if $\cF$ is of degree $1$. Notice that 
a linearly independent factor of degree $1$ is not biased at all, 
and therefore strongly $0$-unbiased. Assume that $\cF$ is not 
strongly $\gamma$-unbiased with respect to the current $\Delta$, 
then there exists $r$ and a set of coefficients 
$\cC=\{c_{i,I}\in \F\}_{i\in [m], I\subseteq [r], |I|\leq \Delta(P_i)}$ for which
$$
\bias(Q_\cC)=\bias\left( \sum_{i\in [m], I\subseteq [r], |I|\leq \Delta(P_i)} c_{i,I}P_i(x+y_I)\right) > \gamma(\dim(\cF)).
$$

To detect this, we will use the following algorithm:
Set $K:=|\F|^{B(r)}$. We can estimate the bias of each of the $K$ possible 
linear combinations $Q_\cC$ and check whether our estimate is greater 
than $\frac{3\gamma(\dim(\cF))}{4}$. Letting $\beta':=\frac{\beta}{4K}$, 
we can distinguish bias $\geq \gamma(\dim(\cF))$ from bias 
$\leq \frac{\gamma(\dim(\cF))}{2}$ correctly with probability $1-\beta'$ 
by computing the average of $e_\F\left(Q_\cC\right)$ on 
$O_{\dim(\cF))}(\frac{1}{\gamma(\dim(\cF))^2 \beta})$ random sets of vectors 
$x, y_1,\ldots,y_r$. Let $\cC$ be such that the estimated bias for $Q_\cC$ 
was greater than $\frac{3\gamma(\dim(\cF))}{4}$, and let $k=\deg(Q_\cC)$. 
By a union bound with probability $1-\frac{\beta}{4}$, 
$\bias(Q_\cC)\geq \frac{\gamma(\dim(\cF))}{2}$, and by \cref{lem:algorithmicBV} 
we can find, with probability $1-\frac{\beta}{4}$, a set of 
polynomials $Q_1,\ldots, Q_s:(\F^n)^{r+1}\rightarrow \F$ of degree $k-1$ such that 
\begin{itemize}
\item $Q_\cC$ is $\frac{\sigma}{2}$-close to a function of $Q_1,\ldots, Q_s$,
\item $s\leq \frac{16 |\F|^5}{\gamma(\dim(\cF))\cdot \sigma \cdot \beta}$.
\end{itemize}
Moreover we know from the 
proof of \cref{lem:algorithmicBV} that for each $j\in [s]$ we have
\begin{align*}
Q_j(x,y_1,\ldots ,y_r)&= Q_\cC(x+h_x, y_1+h_1,\ldots, y_r+h_r)- Q_\cC(x,y_1,\ldots, y_r)\\ &=
\sum_{i\in [m], I\subseteq [r], |I|\leq \Delta(P_i)} c_{i,I}\cdot D_{h_x+ h_I}P_i(x+y_I).
\end{align*}
for some fixed vectors $h_x,h_1,\ldots,h_r\in \F^n$. Let $R^{(j)}_{i,I}\defeq D_{h_x+h_I}P_i$ so that 
$\deg(R^{(j)}_{i,I})< \deg(P_i)$ and $Q_j$ is measurable in 
$\{R^{(j)}_{i,I}(x+y_I)\}_{i\in [m], I\subseteq [r],|I|\leq \Delta(P_i)}$.

Let $P_i$ be a polynomial of maximum degree that appears in $Q_\cC$, 
and let $I$ be maximal such that $c_{i,I}\neq 0$. We will add each 
$R^{(j)}_{i,I}$ to $\cF$ with $\Delta(R^{(j)}_{i,I})=\deg(R^{(j)}_{i,I})$ 
and let $\Delta(P_i):=|I|-1$ and discard $P_i$ if $|I|$ becomes zero. 
Notice that the new factor is $\Delta$-bounded because $P_i(x+y_I)$ 
can be written as a function of $\{Q_j(x,y_1,\ldots,y_r)\}_{j\in [s]}$ and
$\{P_j(x+y_J)\}_{J\subseteq I, |J|\leq \Delta(P_j)}$. Moreover each $Q_j$ is 
measurable in $\{R^{(j)}_{i,I}(x+y_I)\}_{i\in [m], I\subseteq [r], |I|\leq \Delta(P_i)}$.

We can prove that the process stops after a constant number of steps by a strong 
induction on the dimension vector $(M_1,\ldots,M_d)$ of the factor $\cF$, 
where $M_i$ is the number of polynomials of degree $i$ in the factor and 
$d$ is the degree of the factor. The induction step follows from the fact that at 
every step we discard $P_i$ or decrease the $\Delta(P_i)$ 
for some $i$ and add polynomials of lower degree. Let $\cF'$ and 
$\Delta:\cF'\rightarrow \N$ be the output of the algorithm. The 
claim follows by our choices for error distances and the probabilities. 
\end{proof}
\subsection{From Approximation to Computation}\label{sec:approxtocomputelowchar}
Our goal is to prove the following analogue of 
\cref{lem:approxtocomputelowchar} for our notion of strongly unbiased factors.

\begin{lemma}[Approximation by Strongly Unbiased Factor implies Computation]\label[lemma]{lem:approxtocomputelowchar_new}
For every $d\geq 0$ there exists a constant $\sigma_d$ and a decreasing 
function $\gamma:\N\rightarrow \R^+$ such that the following holds. 
Let $P:\F^n\rightarrow \F$ be a polynomial of degree $d$, $\cF=\{P_1,\ldots, P_m\}$ 
be a strongly $\gamma$-unbiased polynomial factor of degree $d-1$ with 
degree bound $\Delta$, and let $\Gamma:\F^m\rightarrow \F$ be a function 
such that $P$ is $\sigma_d$-far from $\Gamma(\cF)$. Then $P$ is in 
fact measurable in $\cF$. 
\end{lemma}

\begin{remark}{\bf(Exact refinement)} 
One can use the above lemma to modify the refinement process of \cref{lem:stronglyunbiasedrefinement} 
to achieve exact refinement instead of approximate refinement. This can be done by using induction on the degree 
of the factor and at every round refining the new polynomials that are to be added to the factor to a strongly unbiased 
set of polynomials. Then \cref{lem:approxtocomputelowchar_new} ensures that the removed polynomial is 
measurable in the new set of polynomials and therefore at every step we have exact refinement of the original 
factor. 
\end{remark}

In order to prove \cref{lem:approxtocomputelowchar_new}, we will try to adapt the proof by~\cite{KL08} to our setting. 
To do this we will need some tools. First lemma which follows immediately 
from parts (1) and (2) of \cref{dfn:stronglyunbiased}.

\begin{lemma}[\cite{KL08}]\label[lemma]{lem:distributionreduction}
Let $\cF=\{P_1,\ldots,P_m\}$ be a $\Delta$-bounded factor. 
Let $x,x'\in \F^n$ be two fixed vectors for which $P_i(x)=P_i(x')$ for all 
$i\in [m]$, namely $x$ and $x'$ belong to the same atom of $\cF$. Let $z_1,\ldots, z_k\in \F^n$ be vectors for some $k\geq 1$, 
and let $Y_1,\ldots, Y_k\in \F^n$ be $k$ random variables. Then the 
following two events are equivalent:
\begin{enumerate}
\item $A=[P_i(x+Y_I)= P_i(x'+z_I) \text{ for all } i\in [m] \text{ and } I\subseteq [k]]$
\item $B=[P_i(x+Y_I)= P_i(x'+z_I) \text{ for all } i\in [m] \text{ and } I\subseteq [k], 1\leq |I|\leq \Delta(P_i)]$.
\end{enumerate}
\end{lemma}

The following equidistribution property of strongly unbiased factors 
will be a main building block of the proof of \cref{lem:approxtocomputelowchar_new}. 

\begin{lemma}[Equidistribution of strongly unbiased factors]\label[lemma]{lem:equidistlowchar}
Let $\gamma:\N\rightarrow \R^+$ be a decreasing function. 
Let $\cF=\{P_1,\ldots,P_m\}$ be a strongly $\gamma$-unbiased factor of degree $d$
with degree bound $\Delta$. Let $Y=(Y_1,\ldots Y_k)\in (\F^n)^k$ be 
random variables. For any non-empty $I\subseteq [k]$, let $x^I\in \F^n$ be 
fixed and $a^{(I)}=(a^{(I)}_1,\ldots,a^{(I)}_k)\in \F^k$ be a vector such that
\begin{itemize}
\item For every $i\in I$, $a^{(I)}_i\neq 0$,
\item For every $i\notin I$, $a^{(I)}_i=0$ 
\end{itemize}
Then the joint distribution of 
$$
\left( P_i\bigg(x^I+\sum_{j\in I}a^{(I)}_j Y_j\bigg): i\in [m], I\subseteq [k], 1\leq |I|\leq \Delta(P_i)\right)
$$
is $\gamma(\dim(\cF))^{1/2^{d}}$-close to the uniform distribution on $\F^{B(k)}$.
\end{lemma}
\begin{proof}
Similar to the proof of \cref{lem:parallelequidistribution}, 
through Fourier analysis it suffices to show that each nonzero linear combination of polynomials
$P_i\big(x^I+\sum_{j\in I}a^{(I)}_j Y_j\big)$ has small bias. Suppose that 
$\cC=\{c_{i,I}\in \F\}_{i\in [m], I\subseteq [k], 1\leq |I|\leq \Delta(P_i)}$ is a collection of 
coefficients, not all of which are zero. Let 
$$
Q_\cC(Y_1,\ldots,Y_k)\defeq \sum_{i\in [m], I\subseteq [k], |I|\leq \Delta(P_i)} c_{i,I}\cdot P_i(x^I+Y_I). 
$$
Let $J$ be a set with maximal support such that there is a nonzero $c_{i,J}$ for some $i\in [m]$. Without 
loss of generality assume that $J=\{1,\ldots, r\}$, for $r=|J|$. We will derive $Q_\cC$ in 
directions of $Y_1,\ldots, Y_r$. Namely, 
$$
DQ(Y_1,\ldots, Y_k, h_1,\ldots, h_r) \defeq D_{h_1}\cdots D_{h_r} Q_\cC(Y_1,\ldots, Y_k),
$$
where $h_i$ is only supported on $Y_i$ coordinates. Then we have
\begin{align*}
DQ(Y_1,\ldots, Y_k, h_1,\ldots, h_r) &= \sum_{i\in [m], I\subseteq [k], |I|\leq \Delta(P_i)}c_{i,I}D_{h_1}\cdots D_{h_r} \big(P_i(x^I+\sum_{j\in I}a^{(I)}_j Y_j)\big) \\ &=
\sum_{i\in [m]: \Delta(P_i)\geq |J|} c_{i,J}\cdot (D_{z_1}\cdots D_{z_r} P_i)(x^J+\sum_{j\in J}a^{(J)}_j Y_j) \\ &=
\sum_{i\in [m]: \Delta(P_i)\geq |J|} c_{i,J}\cdot \sum_{I\subseteq J} (-1)^{|I|} P_i(x^J+\sum_{j\in J}a^{(J)}_j Y_j + z_I),
\end{align*}
where $z_i$ is the projection of $h_i$ on $Y_I$ coordinates. Now by the change of variable $X:=x^J+\sum_{j\in J}a^{(J)}_j Y_j$, we have
$$
DQ(Y_1,\ldots, Y_k, h_1,\ldots, h_r)= \sum_{i\in [m], I\subseteq J, |I|\leq \Delta(P_i)} (-1)^{|I|}c_{i,J} \cdot P_i(X+ z_I).
$$
Now $\cF$ being strongly $\gamma$-unbiased with degree bound $\Delta$ and \cref{claim:repeatedcauchy} imply that
$$
\bias(Q_\cC)^{2^{|J|}} \leq \bias(DQ) \leq \gamma(\dim(\cF)).
$$
\end{proof}

Following is an immediate corollary of \cref{lem:distributionreduction} 
and \cref{lem:equidistlowchar}.

\begin{corollary}\label[corollary]{cor8ofKL}
Suppose that $\cF=\{P_1,...,P_m\}$ is a strongly $\gamma^{2^d}$-unbiased 
factor of degree $d$ with degree bound $\Delta$, for some decreasing function 
$\gamma:\N\rightarrow \R^+$. Let $x,x'\in \F^n$ be two fixed vectors such that 
$P_i(x)=P_i(x')$ for all $i\in [m]$. Let $z_1,\ldots, z_k\in \F^n$ be 
values for some $k\geq 1$, and let $Y_1,\ldots,Y_k\in \F^n$ be $k$ random variables. Then
$$
\Pr \bigg[ P_i(x+Y_I)= P_i(x'+z_I): \forall i\in [m], \forall I\subseteq [k] \bigg]= \big(1\pm \gamma(\dim(\cF))\big) \cdot |\F|^{-B(k)}.
$$
\end{corollary}

The next lemma from \cite{KL08} is the last technical tool we 
will need for the proof of \cref{lem:approxtocomputelowchar_new}. This lemma can be 
adapted to our setting by \cref{cor8ofKL} and 
\cref{lem:equidistlowchar}.

\begin{lemma}[Analogue of \cite{KL08} Lemma~7]\label[lemma]{lem:technicalprobabilities}
Suppose that $\cF$ is a strongly $\gamma^{2^d}$-unbiased factor of degree $d$ with 
degree bound $\Delta$, for a decreasing function $\gamma:\N\rightarrow \R^+$. 
Let $A$ be an atom of $\cF$ and $x\in\F^n$ be a point in $A$. Then:
\begin{enumerate}
\item Let $Y_1,\ldots,Y_{d+1}$ be random variables in $\F^n$. Then
$$
\Pr\big[x+Y_I\in A, \forall I\subseteq [d+1]\big] = \big(1\pm\gamma(\dim(\cF))\big)\cdot |\F|^{-B(d+1)}
$$
\item Let $Y_1,\ldots, Y_{d+1}, Z_1,\ldots, Z_{d+1}$ be random 
variables in $\F^n$. For any non-empty $I_0\in [d+1]$
$$
\Pr\bigg[x+Y_I, x+Z_I\in A, \forall I\subseteq [d+1]\bigg| Y_{I_0}=Z_{I_0} \bigg]\leq \big(1+\gamma(\dim(\cF))\big)|\F|^m \left( |\F|^{-B(d+1)} \right)^2.
$$
\end{enumerate}
\end{lemma}
\begin{proof}
Part (1) follows immediately from \cref{cor8ofKL}. Proof of 
part~(2) almost exactly mimics the proof of Lemma~7 in \cite{KL08}. Their proof (which is a bit
technical) only requires 
equidistribution properties proved in \cref{lem:equidistlowchar},
and thus follows directly from our notion of regularity. We omit the details.
\end{proof}

\begin{proofof}{of \cref{lem:approxtocomputelowchar_new}}
The proof idea is similar to the proof of \cref{prop:approxtocompute_new}. 
Let $\tP\defeq \Gamma(\cF)$. Let $X$ be the set of points in $\F^n$ 
for which $P(x)=\tP(x)$. By assumption we know that 
$|X|\geq (1-\sigma_d)|\F|^n$. We will use $(d+1)$-dimensional parallelepipeds 
to prove that $P(x)=\tP(x)$ within the whole of most of the atoms of $\cF$, 
and then to prove that $P$ should be constant on each remaining atom as well. 

By \cref{lem:equidistribution} we have that for $1-O(\sqrt{\sigma_d})$ 
of the atoms $A$ of $\cF$ we have 
$\Pr_{x\in A}\left(P(x)=\tP(x) \right)\geq 1-O(\sqrt{\sigma_d})$; 
we will say that these atoms are \emph{almost good}. We will prove 
that almost good atoms are in fact, totally good. 

Suppose that $A=\cF^{-1}(t)$ is an almost good atom, where 
$t=(t_{i,j})_{1\leq i \leq d, 1\leq j\leq M_i}$ is a vector in 
$\F^{\dim(\cF)}$. Let $A':= A\cap X$ be the set of points in $A$ 
for which $P(x)=\tP(x)$. We will prove that $A'=A$, namely $P$ is equal to $\tP$ on $A$. 

Let $B=A\backslash A'$ be the set of bad points in the atom $A$. 
Choosing $\sigma_d\leq 2^{-4(d+1)}$, since $A$ is a good atom then 
$|B|\leq 2^{-2(d+1)}|A|$. Assume that $B$ is non-empty. Let 
$Y_1,\ldots, Y_{d+1}$ be random variables in $\F^n$. Choosing $x\in B$, 
by \cref{lem:technicalprobabilities} part (1),
$$
p_A\defeq \Pr\big[x+Y_I\in A, \forall I\subseteq [d+1]\big] \geq \big(1-\gamma(\dim(\cF))\big)\cdot |\F|^{-B(d+1)}.
$$

We will upper bound the event that while all $X+Y_I$ are in $A$, $X+Y_J\in B$ 
for some $J$. To do this, we will apply Cauchy-Schwarz inequality to reduce the 
problem to counting pairs of hypercubes. Fix a non-empty set $I_0\subseteq [d+1]$, and let
\begin{align*}
p_B&\defeq \Pr\big[ x+Y_I \in A, \forall I\subseteq [d+1] \text{ and } x+Y_{I_0}\in B\big]\\&= 
\sum_{x_0\in B} \Pr\big[ x+Y_I\in A, \forall I\subseteq [d+1] \text{ and } x+Y_{I_0}= x_0 \big].
\end{align*}

Let $Z_1,\ldots, Z_{d+1}$ be random variables from $\F^n$ we have
\begin{align*}
p_B^2 &= \left(\sum_{x_0\in B} \Pr\big[ x+Y_I\in A, \forall I\subseteq [d+1] \text{ and } x+Y_{I_0}= x_0 \big] \right)^2\\ &\leq
|B| \sum_{x_0\in B} \Pr\big[ x+Y_I\in A, \forall I\subseteq [d+1] \text{ and } x+Y_{I_0}= x_0 \big]^2 \\ &=
|B|\Pr\bigg[ x+Y_I, x+Z_I\in A, \forall I\subseteq [d+1] \text{ and } x+Y_{I_0}=x+Z_{I_0}  \bigg] \\ &=
|B|\;|\F|^{-n}\Pr\bigg[ x+Y_I,x+Z_I\in A, \forall I\subseteq [d+1] \bigg| x+Y_{I_0}=x+Z_{I_0}\bigg]\\ &\leq
|B|\; |\F|^{m-n} p_A^2\; (1+\gamma(\dim(\cF))),
\end{align*}
where the first inequality is the Cauchy Schwarz inequality and the last 
inequality follows from \cref{lem:technicalprobabilities}, part (2). 
By \cref{lem:equidistribution}, $|A|= \big(1\pm \gamma(\dim(\cF))\big) |\F|^{\dim(\cF)-n}$, thus
$$
p_B^2\leq \frac{|B|}{|A|} p_A^2 \big(1\pm 2\gamma(\dim(\cF))\big) \leq 2^{-2(d+1)}p_R^2,
$$
and thus $\frac{p_B}{p_A}\leq 2^{-(d+1)}\big(1\pm 2\gamma(\dim(\cF))\big)$. 
Now by a union bound over all non-empty $I_0\subseteq [d+1]$ the probability 
that there is some $I_0$ for which $x+Y_{I_0}\in B$ is strictly less than $1$ 
for a small enough $\gamma$. Therefore, there exists $y_1,\ldots,y_{d+1}\in \F^n$ 
for which $x+y_I\in A\backslash B$, for every non-empty $I\subseteq [d+1]$. This 
implies $x\in B$, since $\deg(P)=d$ and therefore
$$
P(x) = \sum_{I\subseteq [d+1], |I|\neq 0} (-1)^{|I|+1}P(x+y_I) = \sum_{I\subseteq [d+1], |I|\neq 0} (-1)^{|I|+1}\tP(x+y_I) = \tP(x).
$$

This proves that every almost good atom is in fact totally good. 
It remains to prove that $P$ is constant on each of the remaining 
$O(\sqrt{\sigma_d})$ of the atoms. Now the proof follows from \cref{claim:badatomsaregood}.
\end{proofof}

\subsection{Applications}
\subsubsection{Computing a biased Polynomial}
Having \cref{lem:stronglyunbiasedrefinement} and 
\cref{lem:approxtocomputelowchar_new} in hand we immediately 
have the following analogue of \cref{thm:biasimplieslowrank}, which 
states that if a polynomial is biased then we can find a factor that computes it. 

\begin{theorem}[Computing a biased polynomial]
\label{thm:biased_compute}
Let $\beta\in(0,1]$ be an error parameter. There is a randomized algorithm 
that given query access to a polynomial $P:\F^n \rightarrow \F$ of degree $d$ 
such that $\bias(P)\geq \delta>0$, runs in $O_{\delta,\beta}(n^d)$ and with 
probability $1-\beta$, returns a polynomial factor $\cF=\{P_1,...,P_{c(d,\delta)}\}$ 
of degree $d-1$ such that $P$ is measurable in $\cF$.
\end{theorem}
\begin{proof}
Let $\sigma_d$ and $\gamma_d$ be as in \cref{lem:approxtocomputelowchar_new}. 
By \cref{lem:algorithmicBV} since $\bias(P)\geq \delta$, 
in time $O_{\sigma_d, \beta}(n^d)$, with probability $1-\frac{\beta}{2}$ 
we can find a polynomial factor $\cF$ of degree $d-1$  such that 
\begin{itemize}
\item $P$ is $\frac{\sigma_d}{2}$-close to a function of $\cF$
\item $\dim(\cF)\leq \frac{4|\F|^5}{\delta^2 \sigma_d \beta}$
\end{itemize}
By \cref{lem:stronglyunbiasedrefinement}, with probability 
$1-\frac{\beta}{2}$, we can find a strongly $\gamma_d$-unbiased factor 
$\cF'$ with degree bound $\Delta$ such that $\cF'$ is $\frac{\sigma_d}{2}$-close 
to being a refinement of $\cF$. 

Thus with probability greater than $1-\beta$, $P$ is $\sigma_d$-close 
to a function of the strongly $\gamma_d$-unbiased factor $\cF'$, and 
it follows from \cref{lem:approxtocomputelowchar_new} that 
$P$ is measurable in $\cF'$. 
\end{proof}

\subsubsection{Worst Case to Average Case Reduction}
Here we will show how \cref{thm:biased_compute} implies an algorithmic 
version of worst case to average case reduction from~\cite{KL08}. To present the 
result, we first have to define what it means for a factor to approximate a 
polynomial.

\begin{definition}[$\delta$-approximation]
We say that a function $f:\F^n \rightarrow \F$ $\delta$-approximates a polynomial $P:\F^n \rightarrow \F$ if 
$$
\left| \E_{x\in\F^n} \left[ e_\F\big(P(x)-f(x)\big)\right]\right|\geq \delta.
$$
\end{definition}

Kaufman and Lovett use \cref{lem:approxtocomputelowchar} to show the following reduction.

\begin{theorem}[Theorem 3 of~\cite{KL08}]
Let $P(x)$ be a polynomial of degree $k$, $g_1,...,g_c$ polynomials of degree $d$, and 
$\Lambda:\F^c\rightarrow \F$ a function such that composition $\Lambda(g_1(x),\ldots,g_c(x))$ 
$\delta$-approximates $P$. Then there exist $c'$ polynomials $h_1,\ldots, h_{c'}$ 
and a function $\Gamma:\F^{c'}\rightarrow \F$ such that
$$
P(x)= \Gamma(h_1(x),\ldots, h_{c'(x)}).
$$
Moreover, $c'=c'(d,c,k)$ and each $h_i$ is of the form $p(x+a)-p(x)$ or $g_j(x+a)$, 
where $a\in \F^n$. In particular, if $k\leq k-1$ then $\deg(h_i)\leq k-1$ also. 
\end{theorem}

Here we will design a randomized algorithm that given $g_1,\ldots,g_k$ 
can compute a set of $h_1,\ldots,h_{c'}$ efficiently. 

\begin{theorem}[Worst-case to average case reduction]
\label{thm:worst_to_average}

Let $\delta, \beta \in (0,1]$ be parameters. There is a randomized algorithm that takes as input
\begin{itemize}
\item A polynomial $P:\F^n\rightarrow \F$ of degree $d$
\item A polynomial factor $\cF=\{P_1,\ldots, P_m\}$ of degree $d-1$
\item A function $\Lambda$ such that $\Lambda(\cF)$ $\delta$-approximates $P$
\end{itemize}
and with probability at least $1-\beta$, returns a polynomial factor 
$\cF'= \{R_1,\ldots,R_{m'}\}$ and a function $\Gamma:\F^{\dim(\cF')}\rightarrow \F$ such that
$$
P(x)= \Gamma(R_1(x),\ldots, R_{m'}(x)),
$$
moreover $c'= O_{\dim(\cF), \beta, \delta, d}(1)$. 
\end{theorem}
\begin{proof}
Looking at the Fourier decomposition of $e_\F(\Lambda)(y_1,\ldots, y_m)$, 
since $\Lambda(P_1,\ldots, P_m)$ $\delta$-approximates $P$, there must exist 
$\alpha=(\alpha_1,\ldots,\alpha_m) \in \F^c$ such that 
$Q_\alpha(x)= \sum_{i\in [m]} \alpha_i\cdot P_i(x)$, $\delta'$-approximates $P$, 
where $\delta'\geq \frac{\delta}{|\F|^m}$. We will estimate $|\bias(P-Q_\alpha))|$ 
for every $\alpha\in \F^m$. For each $\alpha$, we can distinguish 
$\bias(P-Q_\alpha))\leq \frac{\delta'}{2}$ from $\bias(P-Q_\alpha)\geq \delta'$, 
with probability $1-\frac{\beta}{3|\F|^m}$, by evaluating $e_\F(P(x)-Q_\alpha(x))$ 
on $C=O_{\dim(\cF)}\left(\frac{1}{\delta'^2} \log(\frac{1}{\beta})\right)$ random inputs. 
Let $\alpha^*\in\F^m$ be such that our estimate for $\bias(P-Q_{\alpha})$ is greater 
than $\frac{3\delta'}{4}$. 

With probability at least $1-\frac{\beta}{3|\F|^m}$ we will find such 
$\alpha^*$, and by a union bound with probability at least $1-\frac{\beta}{3}$, 
$\bias(P-Q_{\alpha^*})\geq \frac{\delta'}{2}$. Now applying 
\cref{thm:biased_compute} to $P-Q_{\alpha^*}$ with parameters 
$\frac{\beta}{3}$ and $\frac{\delta'}{2}$, we find a polynomial factor $\cF'=\{R_1,\ldots, R_{\tilde{m}}\}$ 
of degree $d-1$, such that with probability $1-\frac{\beta}{3}$, 
$P-Q_{\alpha^*}$ is measurable in $\cF'$. Namely there exists $\Gamma'$ such 
that $P-Q_{\alpha^*}(x)= \Gamma'(\cF'(x))$ and therefore
$$
P(x)= \Gamma'(R_1(x),\ldots, R_{\tilde{m}}(x))+Q_{a^*}(x).
$$
\end{proof}


\bibliographystyle{amsalpha}
\bibliography{regularity}


\end{document}

%% file: macros.tex
\usepackage{bbm}

\usepackage{xspace,xcolor,enumerate}
\usepackage{amsmath,amsfonts,amssymb}
\usepackage{color,graphicx}

\ifnum\showkeys=1
\usepackage[color]{showkeys}
\fi

\definecolor{darkred}{rgb}{0.5,0,0}
\definecolor{darkgreen}{rgb}{0,0.5,0}
\definecolor{darkblue}{rgb}{0,0,0.5}

\usepackage[pdfstartview=FitH,pdfpagemode=None,colorlinks,linkcolor=darkred,filecolor=blue,citecolor=darkred,urlcolor=darkred,pagebackref]{hyperref}

\usepackage{dsfont}

\usepackage[capitalise]{cleveref}

\ifnum\showauthornotes=1
\newcommand{\Authornote}[2]{{\sf\small\color{red}{[#1: #2]}}}
\newcommand{\Authorcomment}[2]{{\sf \small\color{gray}{[#1: #2]}}}
\newcommand{\Authorfnote}[2]{\footnote{\color{red}{#1: #2}}}
\else
\newcommand{\Authornote}[2]{}
\newcommand{\Authorcomment}[2]{}
\newcommand{\Authorfnote}[2]{}
\fi

\ifnum\showdraftbox=1
\newcommand{\draftbox}{\begin{center}
  \fbox{%
    \begin{minipage}{2in}%
      \begin{center}%
        \begin{Large}%
          \textsc{Working Draft}%
        \end{Large}\\
        Please do not distribute%
      \end{center}%
    \end{minipage}%
  }%
\end{center}
\vspace{0.2cm}}
\else
\newcommand{\draftbox}{}
\fi

\newtheorem{theorem}{Theorem}[section]

\newtheorem{definition}[theorem]{Definition}
\newtheorem{lemma}[theorem]{Lemma}

\newtheorem{corollary}[theorem]{Corollary}
\newtheorem{claim}[theorem]{Claim}

\def\FullBox{\hbox{\vrule width 6pt height 6pt depth 0pt}}

\def\qed{\ifmmode\qquad\FullBox\else{\unskip\nobreak\hfil
\penalty50\hskip1em\null\nobreak\hfil\FullBox
\parfillskip=0pt\finalhyphendemerits=0\endgraf}\fi}

\def\qedsketch{\ifmmode\Box\else{\unskip\nobreak\hfil
\penalty50\hskip1em\null\nobreak\hfil$\Box$
\parfillskip=0pt\finalhyphendemerits=0\endgraf}\fi}

\newenvironment{proof}{\begin{trivlist} \item {\bf Proof:~~}}
   {\qed\end{trivlist}}

\newenvironment{proofof}[1]{\begin{trivlist} \item {\bf Proof
#1:~~}}
  {\qed\end{trivlist}}

\def\to{\rightarrow}
\def\eps{\varepsilon}
\def\epsilon{\varepsilon}

\def\phi{\varphi}

\newcommand{\defeq}{\stackrel{\mathrm{def}}=}

\newcommand{\ie}{i.e.,\xspace}

\newcommand{\etal}{et al.\xspace}

\newcommand{\mper}{\,.}
\newcommand{\mcom}{\,,}

\newcommand{\R}{{\mathbb R}}
\DeclareMathOperator*{\E}{\mathbb{E}}
\newcommand{\C}{{\mathbb C}}
\newcommand{\N}{{\mathbb{N}}}

\newcommand{\F}{{\mathbb F}}

\usepackage{nicefrac}

\newcommand{\abs}[1]{\ensuremath{\left\lvert #1 \right\rvert}}

\newcommand{\norm}[1]{\ensuremath{\left\lVert #1 \right\rVert}}

\newcommand{\ip}[1]{\left\langle #1 \right\rangle}

\newcommand{\Esymb}{\mathbb{E}}
\newcommand{\Psymb}{\mathbb{P}}

\DeclareMathOperator*{\ExpOp}{\Esymb}

\DeclareMathOperator*{\ProbOp}{\Psymb}
\renewcommand{\Pr}{\ProbOp}

\newcommand{\Prob}[2]{\Pr_{{#1}}\left[{#2}\right]}

\newcommand{\Ex}[2]{\ExpOp_{{#1}}\left[{#2}\right]}

\newfont{\inhead}{eufm10 scaled\magstep1}
\newcommand{\deffont}{\sf}

\newcommand{\Szemeredi}{Szemer\'edi\xspace}